\documentclass[11pt,a4paper,english]{article} 
\usepackage{authblk}
\title{Quantum Spacetimes  from  General Relativity?}
\author{Albert Much\footnote{much@itp.uni-leipzig.de}}  
\affil{Institut f\"ur Theoretische Physik\\ Universit\"at Leipzig\\ D-04103 Leipzig}  
\usepackage{float}            
\usepackage[T1]{fontenc} 
\usepackage[cmex10]{amsmath}  
\usepackage{amstext}          
\usepackage{mathrsfs}
\usepackage{amsfonts}         
\usepackage{amssymb}          
\usepackage{amsthm}    
\usepackage{setspace}
\usepackage{bm}               
\usepackage{enumerate}        
\usepackage[bottom]{footmisc} 
\usepackage{array}            
\usepackage{textcomp}         
\usepackage{pdfpages}         
\usepackage{parskip}          
\usepackage[right]{eurosym}   
\usepackage{xcolor}           
\usepackage[square,numbers]{natbib}	  
\usepackage{tikz}
\usetikzlibrary{matrix}
\usepackage{bbold}
\usepackage{amsmath}  
\usepackage{tabularx}
\usepackage[colorlinks,pdfpagelabels,pdfstartview = FitH,bookmarksopen = true,bookmarksnumbered =
true,linkcolor = black,plainpages = false,hypertexnames = false,citecolor = black]{hyperref}
\usepackage[a4paper,textwidth=17cm,textheight=24cm,headheight=14pt]{geometry}
\allowdisplaybreaks\usepackage{multirow} 

\newcommand{\mathe}{\mathrm{e}}
\newcommand{\mathi}{\mathrm{i}}
\let\oldre\Re
\let\oldim\Im
\renewcommand{\Re}{\oldre\mathfrak{e}\,}
\renewcommand{\Im}{\oldim\mathfrak{m}\,}
\newcommand{\total}{\mathop{}\!\mathrm{d}}

\newcommand{\eqend}[1]{\,#1}

\newcommand{\ma}{\mathcal{M}}

\newcommand{\gm}{\mathcal{G}_{\mathcal{M}}}

\newtheorem{theorem}{\textsc{Theorem}}[section]

\newtheorem{proposition}[theorem]{\textsc{Proposition}}
\newtheorem{result}[theorem]{\textsc{Result}}

\newtheorem{corollary}[theorem]{\textsc{}Corollary}
\newtheorem{definition}[theorem]{\textsc{Definition}}

\newtheorem{remark}[theorem]{Remark}

\newcommand{\R}{\mathbb{R}}
\newcommand{\C}{\mathbb{C}}

\usepackage{fancyhdr}
\numberwithin{equation}{section} 
\begin{document}
	\maketitle
	\abstract{We introduce a non-commutative product for curved spacetimes, that can be regarded as  a generalization of the Rieffel (or Moyal-Weyl) product. This product employs the exponential map and a Poisson tensor, and   the deformed product maintains associativity   under the condition that the Poisson tensor $\Theta$ satisfies $\Theta^{\mu\nu}\nabla_{\nu}\Theta^{\rho\sigma}=0$, in relation to a Levi-Cevita connection. We proceed to solve the associativity condition for various physical spacetimes, uncovering non-commutative structures with compelling properties.
	}
	\tableofcontents   
	\newpage
	\section{Introduction} 
	Why should we quantize spacetime, i.e.\ what is the necessity to consider a  spacetime with quantum features? While, several arguments exist, let us highlight two fundamental issues, for which a solution lies in the concept of a quantum spacetime.
	
	First,  there is a \textbf{physical limitation} of localization of spacetime points, which stem from the fusion of the quantum-mechanical uncertainty principle with the formation of black holes in general relativity.  In particular, 
	localization with extreme precision induces gravitational collapse, rendering spacetime below the Planck scale devoid of operational significance, see Refs.\cite{Ahluwalia1993,DFR}.
	
	Second, we mention a \textbf{continuity} or transitive argument. If the Einstein equations provide a connection between gravity and spacetime curvature,  any theory of quantum gravity must entail the quantization of spacetime or the spacetime curvature. \par
	A direct   method to conceptualize quantized spacetime or  non-commutative geometry  involves replacing the algebra of smooth functions, denoted as $C^{\infty}(\mathcal{M})$, on a manifold $\mathcal{M}$ with a non-commutative product, see for example the review \cite{SZNCQFT}. This approach parallels the procedure employed in quantum mechanics, where the phase-space algebra is  replaced with a non-commutative product, see \cite{Bayen1, Bayen3, DEFBay}.
	\par
	In addition to addressing the technical challenges inherent in this approach, it would be  desirable to identify equations corresponding to specific spacetimes that yield solutions representing non-commutative structures. This would provide a concrete and systematic framework for understanding the emergence of non-commutative structures in the context of different spacetime geometries analogous to  the metric as a dynamical object in  general  relativity or to the approach of non-commutative geometry via matrix models known as emergent gravity, see \cite{eg1,eg2,eg3}.
	\par
	
	In a recent work, we presented a  deformation quantization framework, following the Rieffel approach, applicable to  globally hyperbolic spacetimes possessing a specific  Poisson structure, \cite{muchboumaye}. Crucially, these Poisson structures must adhere to Fedosov type requirements to ensure   associativity of the resulting deformed product. By applying this  deformation scheme to quantum field theories and their associated states, we established that the deformed state in a non-commutative spacetime exhibits a singularity structure reminiscent of Minkowski spacetime, known as being Hadamard. In particular, we demonstrate that if the undeformed state satisfies the Hadamard condition, then the deformed state also possesses this property.
	
	Within this manuscript, we introduce a definition of  a Rieffel product for curved spacetimes, employing the exponential map.    We  demonstrate that this product satisfies the essential properties of a star product: it is unital, possesses a commutative and flat limit, and exhibits associativity under specific conditions on the Poisson tensor:
	\begin{align}\label{eq:ass}
		\Theta^{\mu\nu}\nabla_{\nu}\Theta^{\rho\sigma}=0.
	\end{align} 
	For a non-degenerate  Poisson tensor this equation reduces to 
	\begin{align}\label{eq:deg}
		\nabla_{\nu}\Theta^{\rho\sigma}=0,
	\end{align}
	which is the Fedosov requirement in disguise. While for a non-degenerate Poisson tensor there are several examples (provided in Section \ref{sec:nondeg}), for various significant spacetimes e.g.\ Schwarzschild or Friedman-Robertson-Walker-Lemaitre,   there are no solutions to the covariant constancy of the Poisson tensor given in Equation \ref{eq:deg}, with exception to the trivial solution $\Theta=0$.  
	\par
	The degenerate case on the other hand offers a bigger variety of solutions and we discuss in this paper  the physical implications of  various non-commutative geometries that we obtain as solutions to this equation.
	\par
	Through our analysis, we establish the viability of this Rieffel product as a fundamental tool in the exploration of non-commutative geometries arising in the context of curved spacetimes.
	\par The paper is structured as follows. In Section \ref{sec:mp} we define the generalized Rieffel product and 
	prove that it defines an associative star product, if the Poisson tensor fulfills Equation \eqref{eq:ass}. Section \ref{sec:assoc} delves into useful mathematical properties of the Poisson tensor. Sections \ref{sec:nondeg} and \ref{sec:degenerate} study solutions to the Poisson tensor for specific spacetimes (i.e.\ solutions to the Einstein equations).  
	\section{Mathematical Preliminaries}\label{sec:mp}
	We begin this section with  a result regarding the existence of     a unique maximal geodesic, \cite[Corollary 4.28]{LeeRiem}.
	\begin{corollary}\label{corexponent}
		Let $\ma$ be a smooth manifold and let $\nabla$ be a connection in its respective tangent bundle  $T\ma$\footnote{A tangent bundle $T\ma$ is the collection of all of the tangent spaces $T_p\ma$ for all points  $p$ on a manifold $\ma$.}. For each $p \in \ma$ and $v \in T_p\ma$, there is a unique maximal geodesic $\gamma : I \rightarrow \ma$ with $\gamma(0) = p$ and $\dot{\gamma}(0) = v$, defined on some open interval $I$ containing $0$.
	\end{corollary}
	This result allows us to define the exponential map, \cite[Chapter 5]{LeeRiem}. The exponential map,  'propagates' the tangent vector \( v \) along a geodesic starting from \( p \) in the direction specified by \( v \), and \( \exp_p(v) \) gives the point reached after moving along this geodesic for a unit parameter length. Hence, the exponential map is the natural   generalization of a translation in flat manifolds to the case of curved manifolds. This is our starting point for the generalization of the Rieffel product from flat to curved.
	
	\begin{definition}
		Define a subset $\mathcal{E} \subseteq T\ma$, \textbf{the domain of the exponential map}, by
		\[ \mathcal{E} = \{ v \in T\ma \, | \, \gamma_v  \text{ is defined on an interval containing } [0,1] \}, \]
		and then define the exponential map $\text{exp} :\mathcal{E} \rightarrow \ma$ by
		\[ \text{exp}(v) = \gamma_v(1), \]
		where $\gamma_v$ is the unique geodesic with initial conditions $\gamma_v(0) = p$ and $\dot{\gamma}_v(0) = v$.
		For each $p \in \ma$, the restricted exponential map at $p$, denoted by $\text{exp}_p$, is the
		restriction of $\text{exp}$ to the set $\mathcal{E}_p = \mathcal{E} \cap T_p\ma$.
	\end{definition} 
	\begin{definition}
		We  denote by $\mathcal{G}_{\ma}$ the set of all transformations that are generated by the exponential map of a manifold $\ma$. This set forms a group, see  \cite{mueck}.
	\end{definition}
	\par Equipped with the former definitions we define the generalized Rieffel product in analogy with \cite{muchboumaye}.  
	
	\begin{definition}
		\label{def:defpro}
		Let the smooth action $\alpha$ of the group $\gm$ denote the geodesic map w.r.t\ the manifold $(\mathcal{M},g)$,   and let $\Theta \in \Gamma^{\infty}(\Lambda^2(T_x\ma))$  be a Poisson bivector.
		Then, the  \textit{formal}  generalized Rieffel product of two functions $f,g\in  C^{\infty}_0(\mathcal{M})$ is defined as
		\begin{align*}
			\label{eq:defpro}
			\left( f \star_\Theta g \right)(x) &\equiv  \lim_{\epsilon \to 0} \iint {\chi}(\epsilon X, \epsilon Y) \, \alpha_{\Theta   X}(f(x))  \, \alpha_{Y}(g(x))   \,\mathe^{- {\mathi} \, (  X,    Y )_x   } \total^N X \total^N Y \\
			&=  \lim_{\epsilon \to 0} \iint \chi(\epsilon X, \epsilon Y) \, f(\exp_{(x)}( \Theta X))  \, g(\exp_{(x)}( Y)) \,  \,\mathe^{- {\mathi} \,X \cdot\, Y  }  \total^N X \total^N Y\eqend{,}
		\end{align*}
		where  $X,Y\in T_x\ma$ and the integrations are w.r.t.\ the non-vanishing components (maximally $dim\ma$)  and  the scalar product $ \cdot  $ is w.r.t.\ the  metric $g_x$   at the point $x$. Moreover,  the cut-off function $\chi \in   C_0^{\infty}(\ma  \times  \ma)$ is chosen such that   condition $\chi(0,0) = 1$ is fulfilled. 
	\end{definition} 
	The main difference between the generalized Rieffel product  and the one introduced in \cite{muchboumaye} is the absence of the embedding formalism.

	This product satisfies the standard properties of a star product, see \cite[Defintion 6.1.1]{waldpoiss} and \cite{Kont}. 
	\begin{proposition}\label{propstar}
		The generalized Rieffel product  fulfills\newline
		\begin{itemize}
			\item Unital, \ $$1\star_\Theta f=f\star_\Theta 1=f$$ \newline
			\item 
			The commutative limit,  $$\displaystyle\lim_{\Theta \rightarrow 0} (f \star_\Theta g) (x)= (f \cdot g) (x),$$  	\newline
			\item The flat limit, i.e.\ in case that the manifold is the four-dimensional   Minkowski  spacetime and a constant  Poisson bivector of maximal rank, the generalized Rieffel product turns to the standard Rieffel product.
		\end{itemize}   
		\end{proposition} 
		\begin{proof}
			To  prove unitality we  use the identity element of the exponential map, i.e.\ $\exp_p(0)=p$, see Corollary \ref{corexponent}.
			
			\begin{align*}
				\left( f \star_\Theta 1 \right)(x) &\equiv \lim_{\epsilon \to 0} \iint {\chi}(\epsilon X, \epsilon Y) \, \alpha_{\Theta   X}(f(x))  \, \alpha_{Y}(1)   \,\mathe^{- {\mathi} \, (  X,    Y )_x  } \total^N X \total^N Y \\
				&=  \lim_{\epsilon \to 0} \iint \chi(\epsilon X, \epsilon Y) \, f(\exp_{(x)}( \Theta X))  \,   \,  \,\mathe^{- {\mathi} \,X \cdot\, Y  }  \total^N X \total^N Y \\
				&=   f(\exp_{(x)}( 0))   \\&=  f(x),
			\end{align*}
			The proof for $\left(1\star_\Theta f \right)(x)$ is analogous. \par Assuming we can interchange the limit and the integrals, which can be done under certain assumptions on the functions $f,g$ see \cite{muchboumaye}, we obtain  
			\begin{align*}
				\lim_{\Theta \rightarrow 0}	\left( f \star_\Theta g \right)(x)  &=
				\lim_{\Theta \rightarrow 0} \iint  \, \alpha_{\Theta X}(f(x)) \, \alpha_{Y}(g(x)) \, \,\mathe^{- {\mathi}  X  \cdot\,Y   }  \total^N X \total^N Y
				\\&=f(x)  \,   \iint  g(\exp_{(x)}( Y))\, \,\mathe^{- {\mathi}  X  \cdot\,Y   }  \total^N X \total^N Y \\&=f(x)\cdot  g(x)
			\end{align*}
			where we used the continuity of $f$ and the property $\exp_{(x)}(0)=x$ of the exponential map. \par
			The exponential map for the Minkowski spacetime is simply a translation, i.e.\ 
			\begin{align*}
				\exp_{(x)}(v)=x+v,
			\end{align*}
			hence we have for the generalized product  the Rieffel product,
			\begin{align*}
				\left( f \star_\Theta g \right)(x) &\equiv \lim_{\epsilon \to 0} \iint {\chi}(\epsilon X, \epsilon Y) \, \alpha_{\Theta   X}(f(x))  \, \alpha_{Y}(g(x))   \,\mathe^{- {\mathi} \, (  X,    Y )_x   } \total^4 X \total^4 Y \\
				&=  \lim_{\epsilon \to 0} \iint \chi(\epsilon X, \epsilon Y) \, f(\exp_{(x)}( \Theta X))  \, g(\exp_{(x)}( Y)) \,  \,\mathe^{- {\mathi} \,X \cdot\, Y  }  \total^4 X \total^4 Y
				\\
				&=  \lim_{\epsilon \to 0} \iint \chi(\epsilon X, \epsilon Y) \, f ( x+\Theta X)  \, g ( x+Y) \,  \,\mathe^{- {\mathi} \,X \cdot\, Y  }  \total^4 X \total^4 Y\eqend{.} 
			\end{align*} 
			
		\end{proof} 
				\label{2.1}
				
			Next, we turn to the question of associativity.

			\begin{theorem}\label{thm:defprodorder}
				Let the Poisson bivector  $\Theta\in \Gamma^{\infty}(\Lambda^2(T_x\ma))$,  fulfill the following condition w.r.t.\ the Levi-Cevita connection 
				\begin{equation} \label{eqpoisassoc}
					\Theta^{\mu\nu}  \nabla_\mu\Theta^{\beta\alpha}	    =0. 
				\end{equation}
				Then, the generalized  Rieffel product  is associative up to second order in $\Theta$, and  is explicitly given by 
				\begin{align*}
					& \left( f \star_\Theta g \right)(x)   =f (x)   g(x) -i \Theta^{\mu\nu}	 \,\partial_\mu f (x) \,\partial_\nu g (x) - \frac{1}{2}  
					\Theta^{\mu\alpha}\Theta^{\nu\beta} 
					\nabla_\mu \partial_\nu f(x) \, \nabla_\alpha \partial_\beta g(x)  +\mathcal{O}(\Theta^3).
				\end{align*} 
				
			\end{theorem}

			\begin{proof}
				See Appendix \ref{appthmassoc}.
				
			\end{proof}
			The former theorem can be as well obtained by a pull-back from the formula supplied in \cite{muchboumaye}, if one uses the exponential map. \begin{remark}
				The deformed product   obtained by the generalized Rieffel product is, up to second order,    equivalent to the star 
				product  given by the  following Drinfeld twist, \cite{drin} in terms of a  formal power series 
				\begin{align}
					\mu_{0}\bigl(   \exp(-i\Theta^{\mu\nu}\nabla_{\mu}\otimes \nabla_{\nu}) (f \otimes g)\bigr),
				\end{align}
				where $\mu_{0}:C^{\infty}(\mathcal{M})\otimes C^{\infty}(\mathcal{M})\rightarrow C^{\infty}(\mathcal{M})$, assuming   the associativity condition given in Equation \eqref{eqpoisassoc}. It is easy to verify that for the flat spacetime and constant $\Theta$, this star product becomes the well-known representation of the Moyal-Weyl product, see \cite{pasc}. 
				
			\end{remark}
			
			\begin{proposition}\label{prop:nc}
				The non-commutative structure between the coordinates, expressed by the star commutator is given to all orders in the deformation matrix $\Theta$ by
				\begin{align}
					[x^{\mu},x^{\nu}]_{\Theta}&:=
					x^{\mu}\star_{\Theta}x^{\nu}-   x^{\nu}\star_{\Theta}x^{\mu} = -2i\Theta^{\mu\nu}.
				\end{align}
			\end{proposition}
			\begin{proof}
				Using the explicit star product we have
				\begin{align*}
					x^{\mu}\star_{\Theta}x^{\nu}  & = x^{\mu}\, x^{\nu}-i\Theta^{\mu\nu}- \frac{1}{2}  
					\Theta^{\kappa\alpha}\Theta^{\lambda\beta} 
					\nabla_\kappa \partial_\lambda x^{\mu} \, \nabla_\alpha \partial_\beta  x^{\nu}  +\mathcal{O}(\Theta^3)\\
					& = x^{\mu}\, x^{\nu}-i\Theta^{\mu\nu}- \frac{1}{2}  
					\Theta^{\kappa\alpha}\Theta^{\lambda\beta}  \, 
					\underbrace{\nabla_\kappa \delta_\lambda^{\mu} }_{=0}\,    \underbrace{\nabla_\alpha \delta_\beta^{\nu}}_{=0}  
				\end{align*}
				and due to the vanishing of the second order term it can be easily seen that all the higher order terms vanish as well.
			\end{proof}
			
			\section{The Associativity Condition}\label{sec:assoc}
			In this section we investigate the mathematical aspects of the  condition of associativity of the generalized star product, i.e.\ 
			\begin{align}\label{eq:equiv1}
				\Theta^{\mu\nu}  \nabla_\mu\Theta^{\beta\alpha}	    =0.
			\end{align} 
			First, we state a result that is essential in regards to this work, \cite[Proposition 1.5]{poison}
			\begin{proposition}
				Let $(\ma, \nabla)$ be a manifold endowed with a torsion-less linear connection. Then, the bivector $\Theta$ defines a Poisson structure on $\ma$ if and only if
				\begin{align}\label{prop15}
					{\Theta}^{\alpha \beta }  \nabla_\alpha {\Theta}^{\mu \nu}+  {\Theta}^{\alpha \mu}  \nabla_\alpha  {\Theta}^{\nu\beta }+  {\Theta}^{\alpha \nu}  \nabla_\alpha  {\Theta}^{\beta \mu}=0.
				\end{align}
			\end{proposition} 
			Next, we give a result that follows from the former proposition.
			
			\begin{proposition}
				A bivector  $\Theta  \in \Gamma^{\infty}(\Lambda^2(T\ma))$ fulfilling the associativity condition \eqref{eq:equiv1} , w.r.t.\ the Levi-Cevita connection, fulfills  Equation \eqref{prop15} and   the Jacobi identity 
				\cite[Equation  1.5]{poison}  
				\begin{align*}
					{\Theta}^{\alpha \beta }  \partial_\alpha {\Theta}^{\mu \nu}+  {\Theta}^{\alpha \mu}  \partial_\alpha  {\Theta}^{\nu\beta }+  {\Theta}^{\alpha \nu}  \partial_\alpha  {\Theta}^{\beta \mu}=0,
				\end{align*} and therefore defines a Poisson structure (or Poisson tensor, see \cite[Definition 4.1.7.]{waldpoiss}) on $\ma$.
			\end{proposition}
			\begin{proof}
				We begin the proof by stating that Equation \eqref{prop15} trivially follows from Condition   \eqref{eq:equiv1}.
				Inserting the Christoffel symbols we obtain
				
				\begin{align*}
					&\Theta^{\alpha \beta}\partial_{\alpha}\Theta^{\mu \nu}+\Theta^{\alpha \beta }\Theta^{\lambda \nu}\,\Gamma^{\mu }_{\alpha \lambda}+\Theta^{\alpha \beta }\Theta^{\mu \lambda}\,\Gamma^{\nu}_{\alpha \lambda}
					\\+& \Theta^{\alpha \mu}\partial_{\alpha}\Theta^{\nu \beta}+\Theta^{\alpha \mu }\Theta^{\lambda \beta}\,\Gamma^{\nu }_{\alpha \lambda}+\Theta^{\alpha \mu }\Theta^{\nu \lambda}\,\Gamma^{\beta}_{\alpha \lambda}\\ 
					+&\Theta^{\alpha \nu}\partial_{\alpha}\Theta^{\beta \mu}+\Theta^{\alpha \nu }\Theta^{\lambda \mu}\,\Gamma^{\beta }_{\alpha \lambda}+\Theta^{\alpha \nu }\Theta^{\beta \lambda}\,\Gamma^{\mu}_{\alpha \lambda} 
					\\&
					{\Theta}^{\alpha \beta }  \partial_\alpha {\Theta}^{\mu \nu}+  {\Theta}^{\alpha \mu}  \partial_\alpha  {\Theta}^{\nu\beta }+  {\Theta}^{\alpha \nu}  \partial_\alpha  {\Theta}^{\beta \mu}
					=0
				\end{align*} 
				where we used  the skew-symmetry of the object  $\Theta$ and
				the symmetry of the Christoffel symbols  to cancel various terms, i.e.\ 
				\begin{align*}
					&  \Theta^{\alpha \beta }\Theta^{\lambda \nu}\,\Gamma^{\mu }_{\alpha \lambda}+  \Theta^{\alpha \nu }\Theta^{\beta \lambda}\,\Gamma^{\mu}_{\alpha \lambda} \\=&
					\Theta^{\alpha \beta }\Theta^{\lambda \nu}\,\Gamma^{\mu }_{\alpha \lambda}+  \Theta^{\lambda \nu }\Theta^{\beta \alpha}\,\Gamma^{\mu}_{ \lambda\alpha} \\=&\Theta^{\alpha \beta }\Theta^{\lambda \nu}\,\Gamma^{\mu }_{\alpha \lambda}-  \Theta^{\lambda \nu }\Theta^{\alpha\beta }\,\Gamma^{\mu}_{ \alpha\lambda} \\=&0.
				\end{align*}
				
			\end{proof}$\,$
			The former proposition serves to assure the tensor structure of $\Theta$ in case that the associativity condition is fulfilled.\par
			
			Assuming non-degeneracy of the Poisson tensor the associativity condition reads  
			\begin{align*}
				\nabla_\mu\Theta^{\beta\alpha}	    =0.
			\end{align*}
			Such a connection is  called a Poisson connection, see \cite{DEFBay} and \cite[Chapter 1.4]{poison}. It is further proven in \cite{poison}  that any Poisson manifold with constant rank Poisson bivector possesses Poisson connections. If the covariant derivative is constant, then the tensor has constant rank and is therefore very close to being "symplectic"\footnote{I am indebted to Stefan Waldmann for this comment.}. In case of maximal rank of the Poisson tensor,  the inverse of $\Theta$ is given by the symplectic structure $\omega$.  For the symplectic case, Fedosov, \cite{fedo1} has proven that a deformation quantization exists if the symplectic structure satisfies
			\cite[Definition 2.3]{fedo1}.
			\begin{align*}
				\nabla\omega=0.
			\end{align*}
			Turning to other literature, let us mention  that in \cite[Defintion 1.1.]{boum}  a triple  $(\mathcal{M},\Theta,g)$ with a Poisson tensor being covariant constant is called  a (Pseudo-)Riemannian Poisson manifold, see \cite[Defintion 1.1.]{boum}. It is a generalization of K\"ahler manifolds.  The covariant constancy of the Poisson tensor appears also in the definition of the Poisson-Riemannian manifold, \cite[Equation 3.1]{majbeg}. See in connection of the covariant constancy of the Poisson tensor \cite{hawk1, hawk2}.
			\par
			In the following sections we assume the case of non-degenerate and degenerate  $\Theta$   hence where Equation \eqref{eq:equiv1} holds. We systematically derive the resulting Poisson structures on a case-by-case basis.

			\section{NC spacetimes - The Non-degenerate Cases}\label{sec:nondeg} 
			In the section we prove   the existence of various Poisson structures fulfilling the condition  of associativity of the generalized star product, i.e.\ 
			\begin{align}\label{eq:equiv}
				\Theta^{\mu\nu}  \nabla_\nu\Theta^{\beta\alpha}	    =0,
			\end{align}  
			w.r.t.\ the Levi-Cevita connection, assuming that $\Theta$ is non-degenerate. 
			\subsection{The Flat Case}
			The metric of Minkowski spacetime $(\ma,\eta)$ in Cartesian coordinates $\{t,x,y,z\in\R\}$ is given by
			\begin{align}
				\boxed{
					\eta=-dt^2+dx^2+dy^2+dz^2 .   
				}
			\end{align}
			For this case we evaluate in the following the   Poisson tensor. This example is of importance due to the following reason: It solidifies the fact that the flat limit reduces to the Rieffel product. In particular,  for the flat case, the  non-degenerate Poisson structure is unique. 
			\begin{result}\label{resmink}
				The non-degenerate Poisson structure, fulfilling Equation \eqref{eq:equiv}, w.r.t.\ the Minkowski spacetime $(\ma,\eta)$ is unique and given by a constant  Poisson structure.
			\end{result}
			\begin{proof}
				The proof is straightforward, since for a flat Levi-Cevita connection, Equation \eqref{eq:equiv} turns to 
				\begin{align*}
					\Theta^{\mu\nu}   \partial_\nu\Theta^{\beta\alpha}	    =0 
				\end{align*}
				together with the non-degeneracy we have 
				\begin{align*}
					\partial_\mu\Theta^{\beta\alpha}	    =0 .
				\end{align*}
			\end{proof}
			Note that in the non-degenerate case, the \textbf{only} solution is a constant Poisson tensor, and thus the flat limit from result \ref{propstar} is satisfied by default. We choose a standard representation of $\Theta$, with the only non-vanishing and constant components $\Theta^{01}=\kappa_e$ and $\Theta^{23}=\kappa_m$, see for example \cite{GL1} or \cite{BLS}.
			\par
			\subsubsection{The Flat Case in Spherical Coordinates} 
			In spherical coordinates $\left\{t\in\R,r\in\R^{+},\vartheta\in(0,\pi),\varphi\in[0,2\pi)\right\}$, the Minkowski metric reads
			\begin{align}
				\boxed{
					\eta=- dt^2 + dr^2 + r^2\left(d\vartheta^2+\sin^2\vartheta d\varphi^2\right) .   
				}
			\end{align} 
			with the following Christoffel symbols, \cite[Section 2.1.3.]{muellercat}
			\begin{subequations}
				\begin{alignat}{5}
					\Gamma_{\vartheta\vartheta}^r &= -r, &\qquad \Gamma_{\varphi\varphi}^r &= -r\sin^2\vartheta, &\qquad \Gamma_{r\vartheta}^{\vartheta} &= \frac{1}{r},\\
					\Gamma_{\varphi\varphi}^{\vartheta} &= -\sin\vartheta\cos\vartheta, & \Gamma_{r\varphi}^{\varphi} &= \frac{1}{r}, & \Gamma_{\vartheta\varphi}^{\varphi} &= \cot\vartheta.
				\end{alignat}
			\end{subequations}
			Either one solves the following differential equations assuming a non-degenerate   $ \Tilde{\Theta}$ that is time-independent, 
			\begin{align*}
				\partial_{0} \Tilde{\Theta}^{\alpha\beta}=0
			\end{align*}
			and the remaining differential equations split in spatial and tempo-spatial part, i.e.\
			\begin{align*}
				\partial_{k} \Tilde{\Theta}^{0i}&=- \Tilde{\Theta}^{0j}\Gamma^{i}_{kj}
				\\
				\partial_{k} \Tilde{\Theta}^{ij}&=- \Tilde{\Theta}^{rj}\Gamma^{i}_{kr} + \Tilde{\Theta}^{ri}\Gamma^{j}_{kr} .
			\end{align*}
			or one uses the tensor property of $ \Tilde{\Theta}$ as follows. Denote by $J$ the Jacobian matrix   of partial derivatives, i.e. 
			\begin{align*}
				J^{\rho}_{\,\,\sigma}=\frac{\partial y^{\rho}}{\partial x^{\sigma}}
			\end{align*}
			where $x$ are the Cartesian coordinates and $y$ represent the spherical coordinates. 
			\begin{result}
				The transformed Poisson tensor $\Tilde{\Theta}$ is given in relation to the matrix $\Theta$ obtained in Result \ref{resmink} by 
				\begin{equation}\label{eqtentra}
					\Tilde{\Theta}^{\rho\sigma}= J^{\rho}_{\,\,\mu} \Theta^{\mu\nu}  J^{\lambda}_{\,\,\nu},
				\end{equation}
				which reads explicitly 
				\begin{align}
					\Tilde{\Theta}^{\rho\sigma}= 
					\left( \begin{matrix}
						0 & \kappa_e\cos(\varphi) \sin(\vartheta) &\kappa_e \frac{\cos(\vartheta) \cos(\varphi)}{r} & -\kappa_e\frac{\csc(\vartheta) \sin(\varphi)}{r} \\
						-\kappa_e\cos(\varphi) \sin(\vartheta) & 0 & -\kappa_m\frac{\sin(\varphi)}{r} & -\kappa_m\frac{\cos(\varphi) \cot(\vartheta)}{r} \\
						-\kappa_e\frac{\cos(\vartheta) \cos(\varphi)}{r} & \kappa_m\frac{\sin(\varphi)}{r} & 0 & \kappa_m\frac{\cos(\varphi)}{r^2} \\
						\kappa_e\frac{\csc(\vartheta) \sin(\varphi)}{r} & \kappa_m\frac{\cos(\varphi) \cot(\vartheta)}{r} & -\kappa_m\frac{\cos(\varphi)}{r^2} & 0
					\end{matrix}\right).
				\end{align}
			\end{result}
			\begin{proof}
				The proof is done via matrix multiplication.
			\end{proof}
			
			This example, though seemingly trivial in its use of tensor properties, is crucial for solidifying the understanding that the representation of the Poisson tensor, and consequently the resulting deformation of the product, responds appropriately to coordinate changes.
			
			\subsubsection{Rindler Wedges}
			A sub-region of the  Minkowski spacetime defined by $|t| < x$,   called the right Rindler wedge,  can be considered as a static globally hyperbolic spacetime, such as Minkowksi, on its own right.
			It describes the spacetime geometry experienced by an accelerating observer in flat spacetime and  is often used in quantum field theory  as a simplified model to study   the Hawking  effect, see Fulling–Davies–Unruh effect \cite{INN6, INN7, INN8}. 
			A convenient coordinate system for the right Rindler wedge is given by, \cite{Crisp}
			\begin{equation}
				t = a^{-1}e^{a\xi}\sinh a\tau ,\;\;\;
				x = a^{-1}e^{a\xi}\cosh a\tau,   
			\end{equation}
			where $a$ is a positive constant. Then, the metric takes the form
			\begin{equation}\boxed{
					g =- e^{2a\xi}(d\tau^2 - d\xi^2) +dy^2 +dz^2.  }
			\end{equation}
			\begin{result}
				Using the tensor property of $\Theta$ we obtain by Equation \eqref{eqtentra} for the components  of the Poisson tensor $\Tilde{\Theta}$ for the right Rindler wedge
				\begin{align}
					\Tilde{\Theta}^{01}=  \kappa_e \,e^{-2a\xi}, \qquad    \Tilde{\Theta}^{23}=  \Theta^{23}=\kappa_m.
				\end{align} 
			\end{result}
			Analog considerations can be done for the left Rindler wedge ($|t| < -x$).  
			\subsection{The   Case $\R^2 \times \mathbb{S}^2$}
			Next, we  search for a solution of Equation \eqref{eq:equiv} for a spacetime that is not entirely flat. We choose the spacetime manifold $(\mathcal{M}=\R^2\times \mathbb{S}^2,g=\eta\oplus g_{\mathbb{S}^2})$, where $g_{\mathbb{S}^2} =d\vartheta^2+\sin^2{(\vartheta)} d\varphi^2$ is the Euclidean metric of a unit two-sphere. Thus, the  full metric   reads 
			\begin{align}\boxed{
					g=-dt^2+dx^2+d\vartheta^2+\sin^2{(\vartheta)} d\varphi^2}
			\end{align}
			with coordinates \{$t,x\in\R$\} and where the angles $\vartheta$ and $\varphi$ are restricted to the ranges  
			$0<\vartheta<\pi$, $0\leq\varphi<2\pi$.
			\begin{result}
				The Poisson tensor w.r.t.\ the manifold $(\mathcal{M}=\R^2\times \mathbb{S}^2,g=\eta\oplus g_{\mathbb{S}^2})$, fulfilling the associativity condition \eqref{eq:equiv}, is given by the components $\Theta^{01}=\kappa_e$,  and $\Theta^{23}=C_1\csc(\vartheta)$, with $C_1\in\R$, with all other components vanishing.
			\end{result}
			
			\begin{proof}
				The proof for the induced metric on $\R^2$ is equivalent to the proof in the  former section. For the sphere we have the following Christoffel symbols
				\begin{align*}
					\Gamma^{2}_{ij}    = \begin{pmatrix}
						0 & 0  \\
						0 & -\sin{(\vartheta)} \cos{(\vartheta)} 
					\end{pmatrix},\qquad \qquad    \Gamma^{3}_{ij} =  \begin{pmatrix}
						0 & \frac{\cos{(\vartheta)} }{\sin{(\vartheta)} }  \\
						\frac{\cos{(\vartheta)} }{\sin{(\vartheta)} }  & 0 
					\end{pmatrix}  
				\end{align*} 
				where $i,j=2,3$. The covariant derivatives read, 
				\begin{align*}
					\nabla_i h &= \partial_i h+\Gamma^{2}_{ik}\Theta^{k3}+\Gamma^{3}_{ik}\Theta^{1k}
					\\&=\partial_i h+\Gamma^{2}_{i1} h+\Gamma^{3}_{i3}h,
				\end{align*}
				where we defined the function $h:= \Theta^{23}$.
				Setting the covariant derivative equal to zero, we have for the derivatives  the following differential equations
				\begin{align*}
					-  \partial_{2} h&= \frac{\cos{(\vartheta)} }{\sin{(\vartheta)} } h
					\\
					-\partial_3 h&=\Gamma^{2}_{32} h+\Gamma^{3}_{33}h=0.
				\end{align*}
				The solution of these equations is given by $h=C_1\csc(\vartheta)$.
			\end{proof}
			If we embed the two-sphere into a higher-dimensional Euclidean space, the derived Poisson tensor takes the following well-known form 
			\begin{align*}
				\Theta^{AB}=\varepsilon^{ABC}X_C, 
			\end{align*}
			where the indices $A,B,C=1,2,3$ and $X$ is the coordinate of the embedding-point on the sphere satisfying 
			\begin{align*}
				X_AX^A=1.
			\end{align*}
			This is known as the fuzzy sphere; see \cite{GRFSPH1, GRFSPH2, GRFSPH3}, \cite[Chapter 7.3]{madore} and references therein.
			\subsection{A simplified Bianchi Universe Type $1$ }
			In this Section we consider the case of a manifold product of a two-dimensional de Sitter spacetime $(d \mathbb{S}^2,g_{d \mathbb{S}^2})$ and a two-dimensional   Euclidean space $(\R^2,\delta )$, which is a special case of the Bianchi Universe Type $1$, where the scale factors in $y$ and $z$ direction are set equal to one. The metric describing this spacetime is thus given by 
			\begin{align*}\boxed{
					g=-dt^2+a(t)^2 dx^2+dy^2+dz^2
				}
			\end{align*}
			where $\{t,x,y,z\in\R\}$ and $a(t)^2$ is the scale factor. 
			\begin{result}
				The Poisson tensor w.r.t.\ the simplified Bianchi Universe Type $1$, given by $(\mathcal{M}=\R^2\times \mathbb{S}^2,g=\eta\oplus g_{\mathbb{S}^2})$, fulfilling the associativity condition \eqref{eq:equiv}, is given by the components $\Theta^{01}=C_1\,a(t)^{-1}$ and $\Theta^{23}=C_2$, where $C_1, C_2\in\R$, with all other components vanishing.
			\end{result}
			
			\begin{proof}
				Due to the simplicity of the metric, the only non-vanishing Christoffel symbols are 
				\begin{align}
					\Gamma^{0}_{11}= a\,\dot{a}, \qquad\qquad \Gamma^{1}_{01}= \Gamma^{1}_{10}= \frac{\dot{a}}{a}.
				\end{align}
				Setting all components of $\Theta$ equal to zero, except for $\Theta^{01}$ and $\Theta^{23}$, the covariant constant condition on the Poisson tensor, i.e.\ $\nabla_{\sigma}\Theta^{\mu\nu}=0$, renders the following differential  equations 
				\begin{align}
					\partial_{i}\Theta^{\mu\nu}&=0,
					\\\nonumber \\
					\partial_{0}\Theta^{01}&=-\Gamma_{01}^1\Theta^{01}  \\\nonumber \\\partial_{0}\Theta^{23}&=0.
				\end{align}
				To which the solutions are $\Theta^{01}=C_1/a(t)$ and $\Theta^{23}=C_2$, with $C_1, C_2\in\R$.
			\end{proof}
			
			\section{NC spacetimes - The Degenerate Cases}\label{sec:degenerate}
			Next, we turn to the case of a degenerate Poisson tensor. The necessity for this requirement stems from the fact that in the presence of a non-vanishing Riemann tensor, the condition of a covariant constant Poisson tensor becomes too restrictive. Specifically, in most four-dimensional spacetimes, this condition results in a vanishing Poisson tensor, which follows from\footnote{I am thankful to M. Fr\"ob for this comment.}
			\begin{align}\label{eqnvcurv}
				[\nabla_{\gamma},\nabla_{\delta}]\Theta^{\alpha\beta}= -R^{\alpha}_{\,\,\, \lambda \gamma\delta}\Theta^{\lambda\beta} -R^{\beta}_{\,\,\, \lambda \gamma\delta}\Theta^{\alpha\lambda} =0.
			\end{align}
			In particular, it leads for most physically relevant spacetimes (e.g.\ Schwarzschild, cosmological spacetimes) in four dimensions to a vanishing Poisson tensor and hence denies the possibility of an associative  deformed product using the geodesic map as given in Definition \ref{def:defpro}. Hence, we consider in this section the possibility of having a degenerate Poisson tensor. This refined approach will remove the previous restrictive limitation of covariant constancy and enable a wider array of solutions within the associativity condition, see Equation \eqref{eqpoisassoc}.
			
			\subsection{Morris-Thorne Wormhole}

			In \cite{morwh} the authors introduced  a traversable wormhole spacetime $(\ma,g_{WH})$, with the goal to be used for 
			rapid interstellar travel, by the following metric,  
			\begin{equation}
				\boxed{
					g_{WH} = -dt^2 + dl^2 + (b_0^2+l^2)\left(d\vartheta^2+\sin^2\!\vartheta\,d\varphi^2\right)}
			\end{equation}
			where $b_0$ is the throat radius and $l$ is the proper radial coordinate; and $\left\{t\in\R,l\in\R,\vartheta\in(0,\pi),\varphi\in [0,2\pi)\right\}$, see Figure \ref{wormholethorne}.\newline\newline
			The Christoffel symbols for this spacetime are given by \cite[Section 2.16]{muellercat}
			\begin{subequations}
				\begin{alignat}{3}
					\Gamma_{12}^{2} &= \frac{l}{b_0^2+l^2}, &\qquad \Gamma_{13}^{3} &= \frac{l}{b_0^2+l^2}, &\qquad \Gamma_{22}^{1} &= -l,\\
					\Gamma_{23}^{3} &= \cot\vartheta, &\qquad \Gamma_{33}^{1} &= -l\sin^2\!\vartheta, & \Gamma_{33}^{2} &= -\sin\vartheta\,\cos\vartheta.
				\end{alignat}
			\end{subequations}
			and the non-vanishing components of the  Riemann   and Ricci tensor are given by
			\begin{align} 
				R_{1212}=-\frac{b_0^2}{b_0^2+l^2},\quad R_{1313}&=-\frac{b_0^2\sin^2\!\vartheta}{b_0^2+l^2},\quad R_{2323}=b_0^2\sin^2\!\vartheta  \\
				R_{11} &= -2\frac{b_0^2}{\left(b_0^2+l^2\right)^2} .
			\end{align}
			\begin{result}
				The associativity condition \eqref{eqpoisassoc} has no solutions for the case of the simple wormhole spacetime $(\ma,g_{WH})$ and  non-degenerate Poisson tensor.  
			\end{result}
			\begin{proof} 
				Using Equation \eqref{eqnvcurv} we have for $\alpha=\gamma$,
				\begin{align*}
					R_{ \lambda  \delta}\Theta^{\lambda\beta} +R^{\beta}_{\,\,\, \lambda \alpha\delta}\Theta^{\alpha\lambda} =0  .
				\end{align*}
				Choosing the indices $\beta=i$ and $\delta=k$ to be spatial
				we have 
				\begin{align*}
					R_{ lk}\Theta^{li} +R^{i}_{\,\,\, ljk}\Theta^{jl} =0   ,
				\end{align*}
				which for $k=2$ gives us 
				\begin{align}\label{eq:wmhrt}
					R^{i}_{\,\,\, lj2}\Theta^{jl} =0   ,
				\end{align}
				and setting $i=1$ or $i=3$ renders 
				\begin{align*}
					R^{1}_{\,\,\, lj2}\Theta^{jl} &= R^{1}_{\,\,\, 212}\Theta^{21}
					=0   ,\\
					R^{1}_{\,\,\, lj3}\Theta^{jl} &= R^{1}_{\,\,\, 313}\Theta^{31}
					=0 
				\end{align*}
				from which we have $\Theta^{21}=\Theta^{31}=0$. Setting in Equation \eqref{eq:wmhrt} the index $k=3$ and choosing $i=2$ renders $\Theta^{23}=0$. Hence, $\Theta^{jk}=0$ eliminates the possibility of a non-degenerate solution. 
			\end{proof}
			We conclude from the former result that we have to search for solutions assuming a degenerate Poison tensor.  
			\begin{figure}
				\centering
				\includegraphics[width=0.5\textwidth]{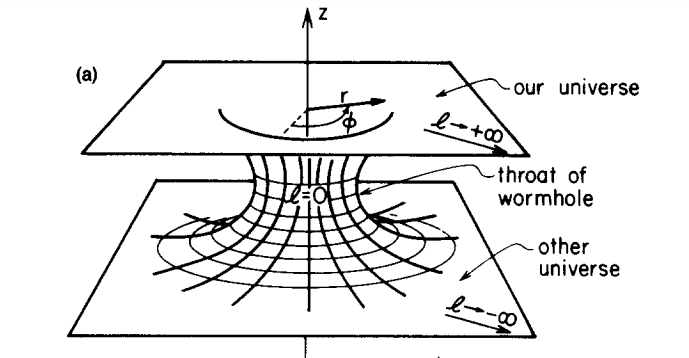}
				\caption{Embedding Diagram for a wormhole that connects two different figures, \cite[Fig. 1.]{morwh}.} 
				\label{wormholethorne}
			\end{figure}
			\begin{result}\label{propwh}
				The Poisson tensor w.r.t.\ the wormhole spacetime $(\ma,g_{WH})$,  fulfilling the associativity condition \eqref{eqpoisassoc}, is given by two-sets of solutions. The first, is given by 
				\begin{align}
					\Theta^{02}&=   \frac{C_1}{{b_0^2 + l^2}},\\
					\Theta^{01}&=\pm\frac{\sqrt{-C_1^2 + 2 b_0^2 C_2 + 2 l^2 C_2}}{\sqrt{b_0^2 + l^2}},
				\end{align}
				with $C_1,\,C_2\in\R$ and all other components are equal to zero. The second solution is given by $\Theta^{\mu\nu}=0$ except for the component $\Theta^{12}$ that has the following explicit  form 
				\begin{align*}
					\Theta^{12}(l)=\frac{C_3}{\sqrt{b_0^2+l^2}},
				\end{align*}
				with $C_3\in\R$.
			\end{result}
			
			\begin{proof}
				{See Appendix \ref{appropwh}.}
			\end{proof}
			For the plot of function $ \Theta^{12}(l)$ (with scale by $C_1/b_0$ and $b_0=1$) see Figure \ref{fig:example}. 
			\begin{figure}
				\centering
				\includegraphics[width=0.5\textwidth]{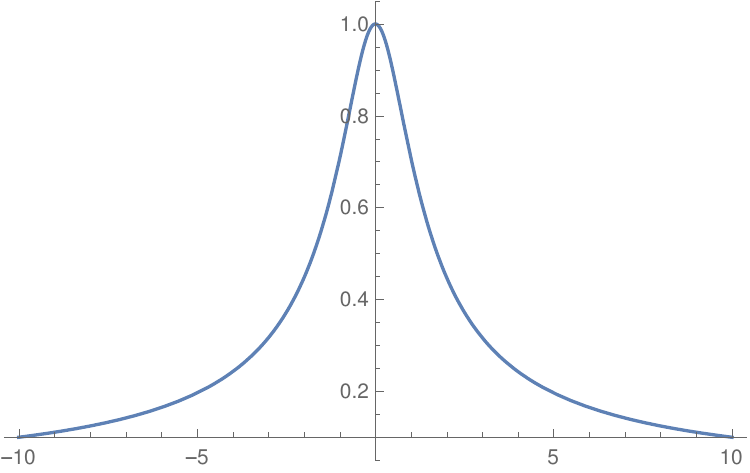}
				\caption{Plot of Function $\Theta^{12}(l)$}
				\label{fig:example}
			\end{figure}The solution states  the following: the non-commutative scale vanishes for large $l$ which is far away from the throat radius but becomes maximal at the throat of the wormhole (i.e.\ $l=0$, see Figure \ref{wormholethorne}). Assuming a quantum structure of the spacetime, this is a natural picture much in analogy with the Planck cell picture of the phase-space. In particular, the tighter the spatial part (namely the circle, where each point is a sphere) of the space-time is squeezed, the larger the non-commutative effect becomes.  The quantum push-back from the commutator relations are larger, the smaller the radius $b_0$ of the worm hole is. This   is consistent with the physical picture of such a quantum spacetime, i.e.\ the tighter the spacetime is confined the larger  non-commutative effects will become.   Since these effects act as a repulsive potential pushing the spacetime walls of the wormhole apart (the closer they are together). These quantum effects might prove fruitful when combating negative energies, of which this spacetime suffers, i.e.\ making rapid interstellar travel possible without negative energy densities. This, of course, has to be proven in the context of     deformed Einstein Equations\footnote{See the last paragraph in Section \ref{sec:conc}.}.  
			
			\subsection{Spherically Symmetric, Static Spacetimes}
			For a large class of static, spherically symmetric spacetimes $(\ma ,g )$, the general form of the metric is given by 
			\begin{align}\label{met:sphsym}\boxed{
					g =-\exp{( 2\alpha(r))}dt^2+\exp{(- 2\alpha(r))} dr^2+r^2\left(
					d\vartheta^2+\sin^2{(\vartheta)} d\varphi^2
					\right)}
			\end{align}
			with $\left\{t\in\R,r\in\R^{+},\vartheta\in(0,\pi),\varphi\in [0,2\pi)\right\}$  and $\alpha(r)$ is a function of the radial component $r$. The non-vanishing  Christoffel symbols for this class of spacetimes are given by 
			
			\begin{equation}
				\begin{matrix}
					\Gamma^1_{00}=\exp{(4\alpha(r) )}\,\partial_1\alpha(r) \,\text{,} & ~ & \Gamma^0_{10} = \partial_1\alpha(r)\,\text{,} & ~ &  \Gamma^1_{11}= -\partial_1\alpha(r)\,\text{,}\\
					~ & ~ & ~ & ~ & ~\\
					\Gamma^1_{22}=-r \exp{( 2 \alpha(r) )}\,\text{,} & ~ & \Gamma^1_{33}=-r \exp{( 2 \alpha(r) )}\sin^2{\vartheta}\,\text{,}  & ~ &  \Gamma^2_{12}=\dfrac{1}{r}\,\text{,}\\
					~ & ~ & ~ & ~ & ~\\
					\Gamma^2_{33}=-\sin{\vartheta}\cos{\vartheta}\,\text{,} & ~ & \Gamma^3_{13}=\dfrac{1}{r}\,\text{,} & ~ & \Gamma^3_{23} = \cot {\vartheta} \,\text{.}\\
					~ & ~ &
				\end{matrix}
			\end{equation} 
			
			\begin{result}\label{propdegeneratesphsym}
				Excluding the case of constant $\alpha$ for the spherically symmetric, static spacetimes the associativity condition \eqref{eqpoisassoc} has no solutions for the case of non-degenerate Poisson tensor.  
			\end{result}
			\begin{proof}
				{See Appendix \ref{appropdegeneratesphsym}.}
				
			\end{proof}
			
			\begin{result}\label{prop:sphsym}
				The Poisson tensor w.r.t.\ class of static, spherically symmetric spacetimes $(\ma,g)$ with the metric tensor of the form given in Equation \eqref{met:sphsym},  fulfilling the associativity condition \eqref{eqpoisassoc}, is given for general $\alpha$ by two-sets of solutions. The first, is given by 
				\begin{align}
					\Theta^{01}&=  C_1,
				\end{align}
				with $C_1$ a real constant  and all other components  equal to zero. The second solution is given by $\Theta^{\mu\nu}=0$ except for the component $f:=\Theta^{12}$ that is the   solution of the following differential equation,  
				\begin{align}\label{solI} 
					\partial_{1}f&=      \left( \partial_1 \alpha(r)   -\frac{1}{r} \right)\, f.
				\end{align} 
				namely
				\begin{align}\label{solII} 
					f(r)&=\frac{C_2 }{r}\exp(\alpha(r))\\
					&=\frac{C_2 }{r} \sqrt{|g_{00}|} ,
				\end{align}
				with $C_2\in\R$.
			\end{result}\begin{corollary}
				In case the function $\alpha(r)$ has the following explicit form 
				\begin{align*}
					\exp(- 2 \alpha(r) )=  {1-C_3e^{z_1}r^2},
				\end{align*}
				where $\{z_1\in\C:e^{z_1}\in\R\}$ and $C_3$ is a constant of spatial dimension $-2$, as for (Anti-)de Sitter in static coordinates,    we have the following set of solutions \begin{align*}
					\Theta^{13}(r)=\exp(  \alpha(r) ) \frac{1}{r} ,\qquad 
					\Theta^{23}(r,\vartheta)&= \exp(-   \alpha(r) ) \frac{\cot{\vartheta}}{r^2}
				\end{align*}
				with all other components vanishing.     
			\end{corollary}
			\begin{proof}
				See Appendix \ref{appropsphsym}.
			\end{proof}
			
			In the following subsections, we investigate the explicit non-commutative structures for spacetimes that are static and spherical symmetric and where the metric takes the form given in Equation \eqref{met:sphsym}.
			
			\subsubsection{Schwarzschild}
			The Schwarzschild spacetime $(\ma_{BH}, g_{BH})$ is a solution to the Einstein field equations of general relativity that describes the spacetime geometry outside a spherically symmetric, non-rotating mass $M$. The metric has the following form, 
			
			\begin{align}\boxed{
					g_{BH}=-\left( 1-\frac{2M}{r} \right)dt^2+\left( 1-\frac{2M}{r} \right)^{-1}dr^2+r^2\left(
					d\vartheta^2+\sin^2{(\vartheta)} d\varphi^2
					\right).}
			\end{align}
			
			\begin{result}\label{prop:sphsymsch}
				The Poisson tensor w.r.t.\ Schwarzschild spacetime $(\ma_{BH}, g_{BH})$,  fulfilling the associativity condition \eqref{eqpoisassoc}, is given by two sets of solutions. The first, is given by 
				\begin{align}
					\Theta^{01}&=  C_1,
				\end{align}
				with $C_1$ being a real constant  and all other components  equal to zero. The second solution is given by $\Theta^{\mu\nu}=0$ except for the component $f:=\Theta^{12}$  that is given by
				\begin{align}\label{eq:schs}
					f(r) = \frac{C_2}{r}  \sqrt{{1 - \frac{2 M}{r}   }} ,
				\end{align}
				where $C_2\in\R$.
			\end{result}
			
			\begin{proof}
				See result \ref{prop:sphsym}.  
			\end{proof}
			If we plot the function $f(r)$ (scaled by $C_2$ and setting $M=1$), see Figure \ref{fig:example2}, we see that  at the event horizon, the non-commutative scale vanishes. The non-commutative strength then inclines and reaches a maximum at $r=3M$ and declines with $r$ inverse.   The radius  $r=3M$ is the inner most stable circular orbit (ISCO) for massive particles. This is the closest radius at which a stable circular orbit is possible around a Schwarzschild black hole. The further we move from the non-rotating mass the more the non-commutative scale becomes negligible, thus demonstrating a deep interconnection between spacetime curvature and non-commutative effects of spacetime.

			\begin{figure}
				\centering
				\includegraphics[width=0.5\textwidth]{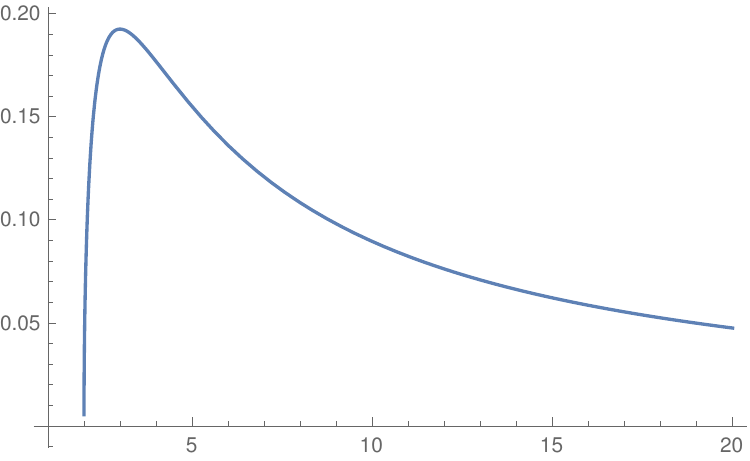}
				\caption{Plot of Function $f(r)$}
				\label{fig:example2}
			\end{figure}

			\subsubsection{Reissner–Nordstr{\o}m}
			The Reissner-Nordstr{\o}m spacetime $(\ma_{RN},g_{RN})$ is a solution to the Einstein field equations that describes the spacetime around a spherically symmetric, electrically charged with charge $Q$, non-rotating mass $M$, with space-time metric 
			\begin{align*}\boxed{
					g_{RN}=-\left( 1-\frac{2M}{r} +\frac{Q^2}{r^2}\right)dt^2+\left( 1-\frac{2M}{r} +\frac{Q^2}{r^2}\right)^{-1}dr^2+r^2\left(
					d\vartheta^2+\sin^2{(\vartheta)} d\varphi^2
					\right).}
			\end{align*}

			\begin{result}\label{prop:sphsymns}
				The Poisson tensor w.r.t.\  Reissner-Nordstr{\o}m  $(\ma_{RN}, g_{RN})$,  fulfilling the associativity condition \eqref{eqpoisassoc}, is given   by two-sets of solutions. The first, is given by 
				\begin{align}
					\Theta^{01}&=  C_1,
				\end{align}
				with $C_1$ being a real constant  and all other components  equal to zero. The second solution is given by $\Theta^{\mu\nu}=0$ except for the component $f:=\Theta^{12}$  that is given by
				\begin{align}\label{eq:solrn}
					f(r) = \frac{C_2 }{r}    {{\sqrt{{1 - \frac{2 M}{r}   +\frac{Q^2}{r^2} }} }} ,
				\end{align}
				where $C_2\in\R$.
			\end{result}

			\begin{proof}
				See result \ref{prop:sphsym}.
				
			\end{proof}
			
			\begin{figure}
				\centering
				\includegraphics[width=0.5\textwidth]{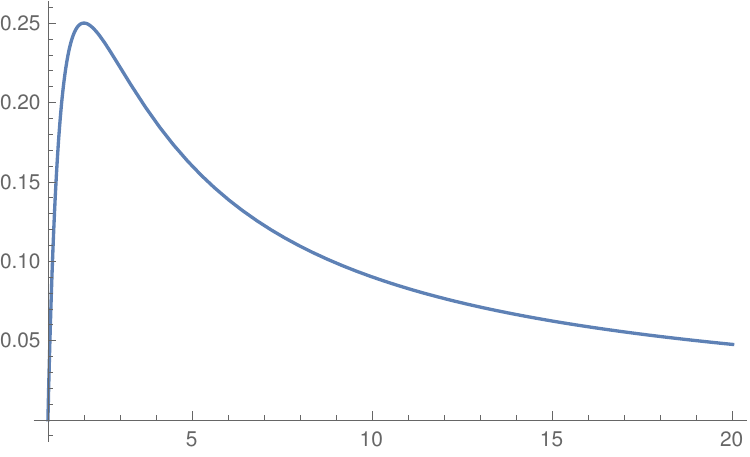}
				\caption{Plot of Function $f(r)$}
				\label{fig:example3}
			\end{figure}
			The interpretation is similar to the uncharged black hole, see Figure \ref{fig:example3}, where we set $M=Q=1$.  
			
			\subsubsection{Kottler or  Schwarzschild-(anti-)deSitter Spacetime}

			The Kottler spacetime $(\ma_{K},g_{K})$ is a solution to the Einstein equation describing the spacetime geometry outside a spherically symmetric mass distribution, with mass $M$, in the presence of a non-zero cosmological constant $\Lambda$. It is represented   by the metric  \cite[Section 2.15]{muellercat}
			\begin{align}\boxed{
					g_K = -\left(1-\frac{2M}{r}-\frac{\Lambda r^2}{3}\right) dt^2+\left(1-\frac{2M}{r}-\frac{\Lambda r^2}{3}\right)^{-1}dr^2 + r^2d\Omega^2 }.
			\end{align} 
			If $\Lambda>0$ the metric is also known as Schwarzschild-de Sitter  metric, whereas if $\Lambda<0$ it is called Schwarzschild-anti-de Sitter.
			
			\begin{result}\label{prop:sphsymnK}
				The Poisson tensor w.r.t.\  Kottler spacetime $(\ma_{K},g_{K})$,  fulfilling the associativity condition \eqref{eqpoisassoc}, is given   by two-sets of solutions. The first, is given by 
				\begin{align}
					\Theta^{01}&=  C_1,
				\end{align}
				with $C_1$ being real constant  and all other components  equal to zero. The second solution is given by $\Theta^{\mu\nu}=0$ except for the component $f:=\Theta^{12}$  that is given by
				\begin{align}\label{eq:solkst}
					f(r) = \frac{C_1}{r}  \sqrt{\left(1-\frac{2M}{r}-\frac{\Lambda r^2}{3}\right)}  
					,
				\end{align}
				where $C_1\in\R$.
			\end{result}
			
			\begin{proof}
				See result \ref{prop:sphsym}. 
			\end{proof}
			
			\subsection{Friedmann-Robertson-Walker-Lemaitre Spacetimes}
			In this section we turn our attention to a  non-static  spacetime, namely  Friedmann-Robertson-Walker-Lemaitre  (\textbf{FRWL}) spacetimes. These spacetimes are a class of cosmological solutions to Einstein's field equations that describe homogeneous and isotropic expanding or contracting universes on large scales. For the following context we consider the FRWL  spacetime  $(\ma_{FRW},g_{FRW})$ for the case of flat spatial geometry, i.e., 
			\begin{align}\boxed{
					g_{FRW}=-dt^2+a^2(t)(dx^2+dy^2+dz^2),}
			\end{align} where $\{t,x,y,z\in\R\}$ and $a(t)^2$ is the scale factor. 
			Assuming  a homogeneous and isotropic spatial part of the spacetime non-commutativity, renders for the Poisson tensor  the following conditions, 
			\begin{align*}
				\partial_i \Theta^{\alpha\beta}=0, \qquad \qquad \Theta^{ij}=0.
			\end{align*}
			This is easily seen, by taking the Lie derivatives of the tensor $\Theta$ along the Killing vector fields of FRWL. Hence, the associativity condition has to be considered for the degenerate case. 
			\begin{result}\label{propfrwl}
				The   Poisson tensor for the FRWL manifold  $(\ma_{FRW},g_{FRW})$, with vanishing spatial derivatives and spatial non-commutativity, i.e.\ $\partial_i \Theta^{\alpha\beta}=0$ and $\Theta^{ij}=0$, is given by 
				\begin{align*}
					\Theta^{0j}= \frac{\Theta}{a(t)}\,e^j,
				\end{align*}
				where $e^j$  is the unit-vector in $j$-direction and $\Theta$ is the Planck-length $\lambda_{p}^2$ squared.
			\end{result}
			
			\begin{proof}
				See Appendix \ref{appropfrwl}.

			\end{proof} 
			The resulting non-commutative structure according to result \ref{prop:nc} is therefore
			\begin{align}\label{eqcrbig}
				[t,x^i]_{\Theta}= 
				-2i\frac{\Theta}{a(t)}\,e^i.
			\end{align}
			A similar commutation relation was obtained in \cite[Theorem II.6]{MF}. Let us apply these commutation relations   heuristically  to the big bang singularity. To make the considerations more clear, we write the commutator relations as the product of uncertainties.  
			\begin{align}\label{eqcrbigunc}
				\Delta_{} T \cdot \Delta_{}  X^i \geq    \lambda_{P}^2 \, \vert \langle  {a(T)^{-1}} \rangle\vert 
			\end{align}
			for all $i$, where  $(T, X^i)$ are the operator representations (assuming they exist) of the commutator relation given in Equation \eqref{eqcrbig}. If one approaches the big bang ($t\rightarrow 0$),  then $a\rightarrow 0$ (see \cite[Page 107]{WA}), assuming the expectation value thereof behaves analogously, the  time  becomes definite. This induces the uncertainty in space to go to infinity. Let us assume that the spacetime fabric is interwoven, i.e.\ no time without space and vice versa. Then, the uncertainty relations \eqref{eqcrbig} suggest  that the occurrence of a big bang singularity can be circumvented by embracing features of a quantum  spacetime. The underlying physical concept is:  attempts to confine the time dimension will induce quantum effects counteracting by exerting pressure in the spatial dimensions, thereby averting the occurrence of a Big Bang singularity (i.e.\ averting a single point). In particular, there will be a minimum scale factor, where the product of the uncertainties in the Inequality \eqref{eqcrbigunc} is equal to the right hand side, implying a minimal size of the universe. 
			The removal of the Big Bang singularity by a non-commutative structure has as well been considered in \cite{bigbang1, bigbang2, bigbang3, bigbang4}. 
			
			\par 
			The inflationary phase, characterized by the scale factor $a(t)=e^{Ht}$, can   be viewed through this noncommutative lens, as well. The smaller $t$, the bigger   quantum effects and the quantum push-back will be in spatial directions. In particular, very early on $t\approx 10^{-43}s-10^{-36}s$, the non-commutative scale will be the largest contributing to   inflation, as a repulsive potential, e.g.\ acting as  an inflaton field. Directly  after inflation the non-commutativity strength  decreases exponentially.
			
			These statements concluded from the non-vanishing commutator relations have to be strengthened by studying deformed Einstein equations or/and QFT in these regimes. This is work in progress.   
			
			\section{Concluding Remarks and Outlook}\label{sec:conc}
			While the associativity condition has supplied us with a plethora of physically meaningful  Poisson tensors corresponding to the specific spacetimes, the condition is still unclear from a physical point of view. In the case of non-degenerate Poisson tensor we can draw, however, some parallels to the metric compatibility condition. In particular $\nabla \pi=0$, means that 
			that the non-commutative commutation relations remain consistent for observers moving along geodesics.
			\par
			In addition to the physical interpretation, another question remains open: How can we establish a connection between the Poisson tensor and the underlying geometry through deformation quantization? In our investigation, we adopted an approach centered on taking a classical spacetime metric. By imposing the associativity condition (stemming from requiring a star product), we derived the non-commutative structure that aligns with this metric. Despite this alignment, these objects (metric and Poisson tensor)  remain distinct entities, lacking a unified framework. Furthermore, our Poisson tensors, although not explicitly stated, carry a deformation parameter of first order, usually identified with the Planck length. Ideally, we want a method that generates corrections to the metric tensor in terms of this deformation parameter and enforcing the associativity condition w.r.t.\ this perturbed metric, should give us a Poisson tensor of next order in the deformation parameter. This approach is  analogous to solving the semi-classical Einstein equations, see \cite{semiclas4, semiclas3, semiclas2, semiclas1} and references therein. 
			\par
			There are (at least) two possibilities of unifying the different geometries given by the metric and Poisson tensor. 
			\par
			The first possibility involves expressing each metric in terms of vielbeins and then employing a generalized deformed product between them, utilizing the evaluated Poisson tensor. This approach, akin to methodologies found in literature such as \cite{bf0, bf1, bf2,   bf, bf4}, typically yields the classical metric alongside perturbations introduced by the deformation parameter.\par
			Another avenue for connecting the metric and the Poisson tensor arises from quantum field theory (QFT) in curved spacetimes, as demonstrated in   \cite{muchboumaye}. Here, we utilize deformation quantization on the two-point function, enabling the definition of a deformed product for pairs of points through geodesic transport. Subsequently, we evaluate the semi-classical Einstein equations with respect to the deformed states. By carefully rearranging the resulting expressions, we can draw conclusions regarding the alteration of geometric quantities, such as the Ricci tensor and scalar. Alternatively, we can interpret the new equations as quantum-corrected semi-classical equations or deformed Einstein equations. This approach offers potential insights into the relationship between the non-commutative structure and the metric and is work in progress.
			
			\section*{Acknowledgements}
			The author thanks R. Verch  for various discussions related to this topic.  We  further extend our thanks to S. Waldmann for comments and  M. Fröb, C. Minz, P. Dorau, R. Ballal, N. L\"ower for various  discussions. Furthermore, we express our gratitude to Z. Avetisyan for  the invitation to the IFS seminar, which greatly clarified and solidified numerous concepts articulated in this paper.
			\appendix
			\section{Proofs}
			\subsection{Proof of Theorem \ref{thm:defprodorder} }\label{appthmassoc}
			\begin{proof}
				
				First, we express the generalized Rieffel product in orders of the deformation matrix

				\begin{align*}
					\left( f \star_\Theta g \right)(x)  
					&=  \lim_{\epsilon \to 0} \iint \chi(\epsilon X, \epsilon Y) \, f(\exp_{(x)}( \Theta X))  \, g(\exp_{(x)}( Y)) \,  \,\mathe^{- {\mathi} \,X \cdot\, Y  }   \eqend{,}
					\\&=\lim_{\epsilon \to 0} \iint \chi(\epsilon X, \epsilon Y) \, (f(x)+( \Theta X)^{\mu}\nabla_{\mu}f+\frac{1}{2} ( \Theta X)^{\mu}( \Theta X)^{\nu}\nabla_{\mu}\nabla_{\nu}f) \, g(\exp_{(x)}( Y)) \,  \,\mathe^{- {\mathi} \,X \cdot\, Y  } 
					\\&=   \lim_{\epsilon \to 0} \iint \chi(\epsilon X, \epsilon Y) \,  f(x) \, g(\exp_{(x)}( Y)) \,  \,\mathe^{- {\mathi} \,X \cdot\, Y  }  
					\\&+  \lim_{\epsilon \to 0} \iint \chi(\epsilon X, \epsilon Y) \, ( \Theta X)^{\mu}(\nabla_{\mu}f(x)) \, g(\exp_{(x)}( Y)) \,\mathe^{- {\mathi} \,X \cdot\, Y  }  
					\\& +\frac{1}{2}  \lim_{\epsilon \to 0} \iint \chi(\epsilon X, \epsilon Y) \, (( \Theta X)^{\mu}( \Theta X)^{\nu}\nabla_{\mu}\partial_{\nu}f(x))\,g(\exp_{(x)}( Y)) \,\mathe^{- {\mathi} \,X \cdot\, Y  }  
					\\&=  f(x) g(x) -i\Theta^{\mu\nu}\partial_{\mu} f(x)\, \partial_{\nu} g(x) 
					\\& +\frac{1}{4}  \nabla_\mu \partial_\nu f \, \nabla_\rho \partial_\sigma g  \,  \lim_{\epsilon \to 0} \iint \chi(\epsilon X, \epsilon Y) \, (( \Theta X)^{\mu}( \Theta X)^{\nu} \,Y^{\rho}Y^{\sigma} \,\mathe^{- {\mathi} \,X \cdot\, Y  }  \\&=  f(x) g(x) -i\Theta^{\mu\nu}\partial_{\mu} f(x)\, \partial_{\nu} g(x) 
					\\& -\frac{1}{4}   \left(   
					\Theta^{\mu\rho}\Theta^{\nu\sigma}
					+\Theta^{\mu\sigma}\Theta^{\nu\rho}
					\right)  \nabla_\mu \partial_\nu f (x) \, \nabla_\rho \partial_\sigma g   (x) 
				\end{align*}
				where in the last lines we used the Taylor expansion of a smooth function of the exponential map, 
				\begin{align*}
					f(\exp_{(x)}( \Theta X))=f(x)+( \Theta X)^{\mu}\nabla_{\mu}f+\frac{1}{2} ( \Theta X)^{\mu}( \Theta X)^{\nu}\nabla_{\mu}\nabla_{\nu}f+\mathcal{O}(\Theta^3)
				\end{align*}
				rewrote terms as differential operators 
				\begin{align*}
					\iint \chi(\epsilon X, \epsilon Y) \, ( \Theta X)^{\mu} \,\mathe^{- {\mathi} \,X \cdot\, Y  }  &=
					\iint \chi(\epsilon X, \epsilon Y) \,  \Theta^{\mu}_{\,\,\nu}X^{\nu} \,\mathe^{- {\mathi} \,X \cdot\, Y  }   
					\\&=i
					\iint \chi(\epsilon X, \epsilon Y) \,  \Theta^{\mu}_{\,\,\nu}g^{\nu\sigma}\partial_{\sigma} \,\mathe^{- {\mathi} \,X \cdot\, Y  }  \\&=i
					\iint \chi(\epsilon X, \epsilon Y) \,  \Theta^{\mu \sigma}\partial_{\sigma} \,\mathe^{- {\mathi} \,X \cdot\, Y  }   
				\end{align*}
				and used 
				partial integration.  Next, note that we have the   relation 
				\begin{align*}
					\nabla_\rho \partial_\sigma g     = \nabla_\sigma \partial_\rho g 
				\end{align*}
				since the function $g$ is smooth and the Christoffel symbols are symmetric in the lower indices. 
				Using this relation we have 
				\begin{align*}
					\Theta^{\mu\sigma}\Theta^{\nu\rho}\nabla_\rho \partial_\sigma g &=
					\Theta^{\mu\sigma}\Theta^{\nu\rho}\nabla_\sigma \partial_\rho g\\&=
					\Theta^{\mu\rho}\Theta^{\nu\sigma}\nabla_\rho \partial_\sigma g,
				\end{align*}
				which turns the star product to 
				\begin{align*}
					\left( f \star_\Theta g \right)(x)  =  f(x) g(x) -i\Theta^{\mu\nu}\partial_{\mu} f(x)\, \partial_{\nu} g(x) 
					-\frac{1}{2}    
					\Theta^{\mu\rho}\Theta^{\nu\sigma} 
					\, \nabla_\mu \partial_\nu f (x) \, \nabla_\rho \partial_\sigma g   (x) .
				\end{align*}
				Next, we prove associativity where we first consider 
				\begin{align*}
					&   ((f\star_{\Theta}g) \star_{\Theta}h ) \,  = (F\star_{\Theta}h) \,  \\&\,\\ & =F   h -i \Theta^{\mu\nu}	 \,\partial_\mu F \,\partial_\nu h  - \frac{1}{2}  
					\Theta^{\mu\alpha}\Theta^{\nu\beta} 
					\nabla_\mu \partial_\nu F\, \nabla_\alpha \partial_\beta h \\&\,\\ & =(f    g -i \Theta^{\mu\nu}	 \,\partial_\mu f  \,\partial_\nu g  - \frac{1}{2}  
					\Theta^{\mu\alpha}\Theta^{\nu\beta} 
					\nabla_\mu \partial_\nu f \, \nabla_\alpha \partial_\beta g )    h \\&\,\\ &-i \Theta^{\mu\nu}	 \,\partial_\mu (f    g 
					-i \Theta^{\mu\nu}	 \,\partial_\mu f  \,\partial_\nu g)  \,\partial_\nu h  
					- \frac{1}{2}  
					\Theta^{\mu\alpha}\Theta^{\nu\beta} 
					\nabla_\mu \partial_\nu (fg) \, \nabla_\alpha \partial_\beta h \\&\,\\ & =(f    g -i \Theta^{\mu\nu}	 \,\partial_\mu f  \,\partial_\nu g  - \frac{1}{2}  
					\Theta^{\mu\alpha}\Theta^{\nu\beta} 
					\nabla_\mu \partial_\nu f \, \nabla_\alpha \partial_\beta g )    h \\&\,\\ &-i \Theta^{\mu\nu}	 ( \,\partial_\mu f \,   g +f     \,\partial_\mu g 
					-i \,\partial_\mu\Theta^{\alpha\beta}	 \,\partial_\alpha f  \,\partial_\beta g -i\Theta^{\alpha\beta}	 \,\,\partial_\mu\partial_\alpha f  \,\partial_\beta g -i\Theta^{\alpha\beta}	 \,\partial_\alpha f  \,\,\partial_\mu\partial_\beta g)  \,\partial_\nu  h  
					\\&\,\\ &- \frac{1}{2}  
					\Theta^{\mu\alpha}\Theta^{\nu\beta} 
					( \nabla_\mu \partial_\nu f\,g+2\partial_\mu f\partial_\nu g +f \,\nabla_\mu \partial_\nu g) \, \nabla_\alpha \partial_\beta h 
				\end{align*} 
				
				Next, we consider 
				\begin{align*}
					&   (f\star( g \star_{\Theta}h )) \,  = (f \star_{\Theta}G) \, 
					\\&\,\\ &= 
					f   G -i \Theta^{\mu\nu}	 \,\partial_\mu f  \,\partial_\nu G  - \frac{1}{2}  
					\Theta^{\mu\alpha}\Theta^{\nu\beta} 
					\nabla_\mu \partial_\nu f \, \nabla_\alpha \partial_\beta G \\&\,\\ &= 
					f   (gh -i \Theta^{\mu\nu}	 \,\partial_\mu g  \,\partial_\nu h  - \frac{1}{2}  
					\Theta^{\mu\alpha}\Theta^{\nu\beta} 
					\nabla_\mu \partial_\nu g \, \nabla_\alpha \partial_\beta h)\\&\,\\ & 
					-i \Theta^{\mu\nu}	 \,\partial_\mu f  \,\partial_\nu (gh -i \Theta^{\alpha\beta}	 \,\partial_\alpha g  \,\partial_\beta h)  
					- \frac{1}{2}  
					\Theta^{\mu\alpha}\Theta^{\nu\beta} 
					\nabla_\mu \partial_\nu f \, \nabla_\alpha \partial_\beta (gh)\\&\,\\ &= 
					f   (gh -i \Theta^{\mu\nu}	 \,\partial_\mu g  \,\partial_\nu h  - \frac{1}{2}  
					\Theta^{\mu\alpha}\Theta^{\nu\beta} 
					\nabla_\mu \partial_\nu g \, \nabla_\alpha \partial_\beta h)\\&\,\\ & 
					-i \Theta^{\mu\nu}	 \,\partial_\mu f  ( \,\partial_\nu g\,h+g \,\partial_\nu h -i  \,\partial_\nu\Theta^{\alpha\beta}	 \,\partial_\alpha g  \,\partial_\beta h-i \Theta^{\alpha\beta}	 \, \,\partial_\nu\partial_\alpha g  \,\partial_\beta h-i \Theta^{\alpha\beta}	 \,\partial_\alpha g  \, \,\partial_\nu\partial_\beta h)  
					\\&\,\\ &    - \frac{1}{2}  
					\Theta^{\mu\alpha}\Theta^{\nu\beta} 
					\nabla_\mu \partial_\nu f  (\, \nabla_\alpha \partial_\beta g\,h+2\partial_\alpha g\partial_\beta h +g\, \nabla_\alpha \partial_\beta h)
				\end{align*}
				comparing the two expressions $f\star( g \star_{\Theta}h )$ and  $(f\star  g) \star_{\Theta}h $
				we have,
				\begin{align*}
					&  -i \Theta^{\mu\nu}	 (  
					-i \,\partial_\mu\Theta^{\alpha\beta}	 \,\partial_\alpha f  \,\partial_\beta g -i\Theta^{\alpha\beta}	 \,\,\partial_\mu\partial_\alpha f  \,\partial_\beta g  )  \,\partial_\nu  h  
					-  
					\Theta^{\mu\alpha}\Theta^{\nu\beta} 
					\,   \partial_\mu f \,\partial_\nu g    \, \nabla_\alpha \partial_\beta h 
					\\&\,\\ &=  
					-i \Theta^{\mu\nu}	 \,\partial_\mu f  (   -i  \,\partial_\nu\Theta^{\alpha\beta}	 \,\partial_\alpha g  \,\partial_\beta h -i \Theta^{\alpha\beta}	 \,\partial_\alpha g  \, \,\partial_\nu\partial_\beta h)    -  
					\Theta^{\mu\alpha}\Theta^{\nu\beta} 
					\nabla_\mu \partial_\nu f   \, \partial_\alpha g \,\partial_\beta h   
				\end{align*}
				which summarizes to 
				\begin{align*}
					&  -  \Theta^{\mu\nu}	 (  
					\,\partial_\mu\Theta^{\alpha\beta}	 \,\partial_\alpha f  \,\partial_\beta g   )  \,\partial_\nu  h  
					-  
					\Theta^{\mu\nu}\Theta^{\alpha\beta}  \, \Gamma_{\nu\beta}^\gamma 
					\,   \partial_\mu f \,\partial_\alpha g   \, \partial_\gamma h 
					\\&\,\\ &=  
					-  \Theta^{\mu\nu}	 \,\partial_\mu f  (      \,\partial_\nu\Theta^{\alpha\beta}	 \,\partial_\alpha g  \,\partial_\beta h  )      -  
					\Theta^{\mu\nu}\Theta^{\alpha\beta} \,
					\Gamma_{\mu\alpha}^{\gamma}\,\partial_\gamma f   \, \partial_\nu g \,\partial_\beta h   
				\end{align*}
				using the Jacobi identity we obtain 
				\begin{align*}
					\Theta^{\mu\nu}  \,\partial_\mu\Theta^{\beta\alpha}	 +
					\Theta^{\mu\nu} \Theta^{\alpha\delta}  \, \Gamma_{\mu\delta}^\beta   +  \Theta^{\mu\nu}\Theta^{\delta\beta} \,
					\Gamma_{\mu\delta}^{\alpha}   =0
				\end{align*} 
				which can be written as the vanishing of the covariant derivative 
				\begin{align*}
					\Theta^{\mu\nu}  \,\nabla_\mu\Theta^{\beta\alpha}	    =0.
				\end{align*} 
				
			\end{proof}
			
			\subsection{Proof of Result \ref{propwh}}\label{appropwh}
			
			\begin{proof}
				
				The associativity condition for    $\alpha=0$ and $\beta=i$ reads
				\begin{align*}
					\Theta^{\mu\nu}\partial_{\nu}\Theta^{0i}= -\Theta^{\mu j}\Theta^{0k}\,\Gamma^{i}_{jk}.
				\end{align*}
				Assuming $\partial_{0,3}\Theta=0$ and using the notation  $a,b=1,2$ we have
				\begin{align*}
					\Theta^{\mu a}\partial_{a}\Theta^{0i}= -\Theta^{\mu j}\Theta^{0k}\,\Gamma^{i}_{jk}.
				\end{align*}
				We have for $i=1$
				\begin{align*}
					\Theta^{\mu a}\partial_{a}\Theta^{01}&= -\Theta^{\mu j}\Theta^{0k}\,\Gamma^{1}_{jk}
					\\&=-\Theta^{\mu 2}\Theta^{02}\,\Gamma^{1}_{22}-\Theta^{\mu 3}\Theta^{03}\,\Gamma^{1}_{33}
				\end{align*}
				while for $i=2$
				\begin{align*}
					\Theta^{\mu a}\partial_{a}\Theta^{02}&= -\Theta^{\mu j}\Theta^{0k}\,\Gamma^{2}_{jk}
					\\&=-   \Theta^{\mu 1}\Theta^{02}\,\Gamma^{2}_{12} -\Theta^{\mu 2}\Theta^{01}\,\Gamma^{2}_{21} -\Theta^{\mu 3}\Theta^{03}\,\Gamma^{2}_{33}
				\end{align*}
				while for $i=3$
				\begin{align*}
					\Theta^{\mu a}\partial_{a}\Theta^{03}&= -\Theta^{\mu j}\Theta^{0k}\,\Gamma^{3}_{jk}
					\\&= -\Theta^{\mu 1}\Theta^{03}\,\Gamma^{3}_{13}-\Theta^{\mu 3}\Theta^{01}\,\Gamma^{3}_{31}-\Theta^{\mu 2}\Theta^{03}\,\Gamma^{3}_{23}-\Theta^{\mu 3}\Theta^{02}\,\Gamma^{3}_{32}
				\end{align*}
				For $\alpha=i$ and $\beta=j$ we have 
				
				\begin{align*}
					\Theta^{\mu a}\partial_{a}\Theta^{ij}= - \Theta^{\mu k}\Theta^{sj }\,\Gamma^{i}_{ks}-\Theta^{\mu k}\Theta^{is}\,\Gamma^{j}_{ks}
				\end{align*}
				
				One set of solutions is given by setting $\Theta^{ik}=0$ and $\Theta^{03}=0$ rendering the following set of differential equations 
				\begin{align*}
					\Theta^{0 a}\partial_{a}\Theta^{01}&=-(\Theta^{02})^2 \,\Gamma^{1}_{22}
					\\  \Theta^{0 a}\partial_{a}\Theta^{02}&= - 2  \Theta^{0 1}\Theta^{02}\,\Gamma^{2}_{12} 
				\end{align*}
				Further assuming that the functions do not depend on the angle $\vartheta$ we have 
				\begin{align} 
					\Theta^{0 1}\partial_{1}\Theta^{01}&=-(\Theta^{02})^2 \,\Gamma^{1}_{22}
					\\   \partial_{1}\Theta^{02}&= - 2   \Theta^{02}\,\Gamma^{2}_{12} 
				\end{align}
				to which the solutions are 
				\begin{align}
					\Theta^{02}&=   \frac{C_1}{{b_0^2 + l^2}}\\
					\Theta^{01}&=\pm\frac{\sqrt{-C_1^2 + 2 a^2 C_2 + 2 l^2 C_2}}{\sqrt{a^2 + l^2}}
				\end{align}
				Setting the dimension-full constant $C_1=0$ we have a constant non-commutativity between the temporal and radial component. 
				\par
				The next set of solutions we consider is $\Theta^{0k}=0$ and $\Theta^{13}=\Theta^{23}=0$, rendering 
				\begin{align*}
					\Theta^{\mu a}\partial_{a}\Theta^{12}& =  -\Theta^{\mu 1}\Theta^{12}\,\Gamma^{2}_{12} 
				\end{align*} 
				which is a set of two differential equations
				\begin{align*}
					\partial_{1}\Theta^{12}& =  - \Theta^{12}\,\Gamma^{2}_{12} \\
					\partial_{2}\Theta^{12}& = 0,
				\end{align*}
				with solution 
				\begin{align*}
					\Theta^{12}(l)=\frac{C_1}{\sqrt{b_0^2+l^2}}
				\end{align*}
			\end{proof}
			\subsection{Proof of Result \ref{propdegeneratesphsym}}\label{appropdegeneratesphsym}
			\begin{proof}
				
				The components of the Riemann   and Ricci tensor are given by  
				\begin{align*}
					& R_{ \,\,101}^0= \left( -\partial_1^2 \alpha-2\left(\partial_1 \alpha\right)^2\right) , \quad  R_{\, \,202}^0=-r e^{ 2 \alpha} \partial_1 \alpha, \quad
					R_{\, \,303}^0=-r e^{ 2 \alpha} \partial_1 \alpha \,\sin ^2 \vartheta  \\ \nonumber\\
					& R_{\, \,212}^1=-r e^{ 2 \alpha} \partial_1 \alpha,  \quad R_{\, \,313}^1=-r e^{ 2 \alpha} \partial_1 \alpha \sin ^2 \vartheta ,  \quad  R_{\,\, 323}^2=\left(1-e^{ 2 \alpha}\right) \sin ^2 \vartheta \\ \nonumber\\  
					&R_{00}   = e^{4\alpha}\left(\partial_1^2 \alpha+2\left(\partial_1 \alpha\right)^2 +\frac{2}{r} \partial_1 \alpha\right), \quad
					R_{11}   =-\left(\partial_1^2 \alpha+2\left(\partial_1 \alpha\right)^2 +\frac{2}{r} \partial_1 \alpha\right) , \\ \nonumber\\ 
					&R_{22}   =-e^{ 2 \alpha}\left(r\left( 2\partial_1 \alpha\right)+1\right)+1 , \quad
					R_{33}   =R_{22} \sin ^2 \vartheta
				\end{align*}       
				Using Equation \eqref{eqnvcurv} we have for $\alpha=\gamma$
				\begin{align*}
					R_{ \lambda  \delta}\Theta^{\lambda\beta} +R^{\beta}_{\,\,\, \lambda \alpha\delta}\Theta^{\alpha\lambda} =0  .
				\end{align*}
				Using the  explicit expressions for the components we have
				\begin{center}
					\begin{tabular}{ |c|c|c|c| } 
						\hline
						$\beta$ & $\delta$ &   Solution \\
						\hline
						{0} & 1 &   $\Theta^{01}=0$ \\ 
						0 & 2 &  $\Theta^{02}=0$ or Solution \eqref{eq2}\\ 
						0 & 3 & $\Theta^{03}=0$ or Solution \eqref{eq2} \\  
						1 & 2 & $\Theta^{12}=0$ or Solution \eqref{eq2} \\  
						1 & 3 & $\Theta^{13}=0$ or Solution \eqref{eq2} \\  
						2 & 0 & $\Theta^{02}=0$ or Schwarzschild Solution $\rightarrow$  $\Theta^{02}=0$\\  
						2 & 1 & $\Theta^{12}=0$ or Schwarzschild Solution  $\rightarrow$  $\Theta^{12}=0$ \\  
						2 & 3 & $\Theta^{23}=0$  \\  
						3 & 1 & $\Theta^{13}=0$ or Schwarzschild Solution  $\rightarrow$  $\Theta^{13}=0$ \\  
						\hline
				\end{tabular}\end{center}$\,$\newline
				where we have the differential equation 
				\begin{equation}\label{eq21}
					r\partial_1\alpha+1-e^{-2\alpha}=0
				\end{equation}
				with solution 
				\begin{equation}\label{eq2}
					\alpha(r)=\frac{1}{2}\log( 1-C_1/r^2)
				\end{equation}
				The only non-vanishing component is  $\Theta^{03}$ for the Solution \eqref{eq2}, rendering a degenerate $\Theta$.
				
			\end{proof}

			\subsection{Proof of Result \ref{prop:sphsym}}\label{appropsphsym}
			\begin{proof}

				Assuming $\partial_{0,3}\Theta=0$ and using the indices $a,b=1,2$ we have for the associativity condition \eqref{eq:equiv}, 
				
				\begin{align*}
					\Theta^{\mu a}\partial_{a}\Theta^{0 k} &=- \Theta^{\mu0}\Theta^{1 
						k}\,\Gamma^{0}_{01}- \Theta^{\mu1}\Theta^{0
						k}\,\Gamma^{0}_{10}-\Theta^{\mu j}\Theta^{0i}\,\Gamma^{ k}_{j i}
				\end{align*}
				This gives us 
				
				\begin{align*}
					\Theta^{\mu a}\partial_{a}\Theta^{01} &=  - \Theta^{\mu1}\Theta^{0
						1}\,\Gamma^{0}_{10}-\Theta^{\mu j}\Theta^{0i}\,\Gamma^{ 1}_{j i}
					\\&=- \Theta^{\mu1}\Theta^{0
						1}\,\Gamma^{0}_{10}-\Theta^{\mu 1}\Theta^{01}\,\Gamma^{ 1}_{11}-\Theta^{\mu 2}\Theta^{02}\,\Gamma^{ 1}_{22}-\Theta^{\mu 3}\Theta^{03}\,\Gamma^{ 1}_{33} \\&= -\Theta^{\mu 2}\Theta^{02}\,\Gamma^{ 1}_{22}-\Theta^{\mu 3}\Theta^{03}\,\Gamma^{ 1}_{33}
				\end{align*}
				since $\Gamma^{0}_{10}+\Gamma^{ 1}_{11}=0$,
				\begin{align*}
					\Theta^{\mu a}\partial_{a}\Theta^{0 2} &=- \Theta^{\mu0}\Theta^{1 
						2}\,\Gamma^{0}_{01}- \Theta^{\mu1}\Theta^{0
						2}\,\Gamma^{0}_{10}-\Theta^{\mu j}\Theta^{0i}\,\Gamma^{ 2}_{j i}
					\\ &=- \Theta^{\mu0}\Theta^{1 
						2}\,\Gamma^{0}_{01}- \Theta^{\mu1}\Theta^{0
						2}\,\Gamma^{0}_{10}-\Theta^{\mu3}\Theta^{03}\,\Gamma^{ 2}_{33}-\Theta^{\mu 1}\Theta^{02}\,\Gamma^{ 2}_{12}-\Theta^{\mu 2}\Theta^{01}\,\Gamma^{ 2}_{21}
				\end{align*}
				\begin{align*}
					\Theta^{\mu a}\partial_{a}\Theta^{0 3} &=- \Theta^{\mu0}\Theta^{1 
						3}\,\Gamma^{0}_{01}- \Theta^{\mu1}\Theta^{0
						3}\,\Gamma^{0}_{10}-\Theta^{\mu j}\Theta^{0i}\,\Gamma^{ 3}_{j i}\\
					&=- \Theta^{\mu0}\Theta^{1 
						3}\,\Gamma^{0}_{01}- \Theta^{\mu1}\Theta^{0
						3}\,\Gamma^{0}_{10}-\Theta^{\mu 1}\Theta^{03}\,\Gamma^{ 3}_{13}-\Theta^{\mu 3}\Theta^{01}\,\Gamma^{ 3}_{31}-\Theta^{\mu 2}\Theta^{03}\,\Gamma^{ 3}_{23}-\Theta^{\mu3}\Theta^{02}\,\Gamma^{ 3}_{32}
				\end{align*}
				Second, we consider the terms where $\alpha=i,\,\beta=k$
				
				\begin{align*}
					\Theta^{\mu a}\partial_{a}\Theta^{i k}&= - \Theta^{\mu\nu}\Theta^{\lambda 
						k}\,\Gamma^{i}_{\nu\lambda}-\Theta^{\mu\nu}\Theta^{i\lambda}\,\Gamma^{ k}_{\nu\lambda}
				\end{align*}
				rendering 
				
				\begin{align*}
					\Theta^{\mu a}\partial_{a}\Theta^{1 2}&= - \Theta^{\mu\nu}\Theta^{\lambda 
						2}\,\Gamma^{1}_{\nu\lambda}-\Theta^{\mu\nu}\Theta^{1\lambda}\,\Gamma^{ 2}_{\nu\lambda}
					\\&=- \Theta^{\mu0}\Theta^{0 2}\,\Gamma^{1}_{00}- \Theta^{\mu1}\Theta^{1 2}\,\Gamma^{1}_{11} - \Theta^{\mu3}\Theta^{32}\,\Gamma^{1}_{33}
					-\Theta^{\mu3}\Theta^{1 3}\,\Gamma^{ 2}_{33}-\Theta^{\mu1}\Theta^{1 2}\,\Gamma^{ 2}_{12} 
				\end{align*}

				\begin{align*}
					\Theta^{\mu a}\partial_{a}\Theta^{1 3}&= - \Theta^{\mu\nu}\Theta^{\lambda 
						3}\,\Gamma^{1}_{\nu\lambda}-\Theta^{\mu\nu}\Theta^{1\lambda}\,\Gamma^{ 3}_{\nu\lambda}
					\\&= - \Theta^{\mu0}\Theta^{0
						3}\,\Gamma^{1}_{00}- \Theta^{\mu1}\Theta^{1
						3}\,\Gamma^{1}_{11}- \Theta^{\mu2}\Theta^{2
						3}\,\Gamma^{1}_{22}-\Theta^{\mu1}\Theta^{1 3}\,\Gamma^{ 3}_{13}-\Theta^{\mu2}\Theta^{1 3}\,\Gamma^{ 3}_{23}-\Theta^{\mu3}\Theta^{1 2}\,\Gamma^{ 3}_{32}
				\end{align*} 
				
				\begin{align*}
					\Theta^{\mu a}\partial_{a}\Theta^{2 3}&= - \Theta^{\mu\nu}\Theta^{\lambda 
						3}\,\Gamma^{2}_{\nu\lambda}-\Theta^{\mu\nu}\Theta^{2\lambda}\,\Gamma^{ 3}_{\nu\lambda}\\&
					= - \Theta^{\mu1}\Theta^{2
						3}\,\Gamma^{2}_{12}- \Theta^{\mu2}\Theta^{1
						3}\,\Gamma^{2}_{21}-\Theta^{\mu1}\Theta^{2 3}\,\Gamma^{ 3}_{13}-\Theta^{\mu3}\Theta^{21}\,\Gamma^{ 3}_{31} -\Theta^{\mu2}\Theta^{2 3}\,\Gamma^{ 3}_{23}
				\end{align*}

				One solution is given by setting all components equal to zero except $\Theta^{01}$ rendering the only remaining differential equation \begin{align*}
					\partial_1\Theta^{01}=0,
				\end{align*}with solution $\Theta^{01}=C_{01}$, with $C_{01}\in\R$.
				A different set of solutions is given if we set $\Theta^{0i}$ equal to zero. Then, the differential equations reduce to

				\begin{align} \label{eqdiffthet1}
					f\partial_{2}f&=     gh\,\Gamma^{1}_{33}
					-g^2\,\Gamma^{ 2}_{33} 
				\end{align}
				
				\begin{align} \label{eqdiffthet2}
					-   f\partial_{1}f&=    f^2\,\Gamma^{1}_{11} + h^2\,\Gamma^{1}_{33}
					-hg\,\Gamma^{ 2}_{33}+f^2\,\Gamma^{ 2}_{12} 
				\end{align}
				
				\begin{align} \label{eqdiffthet3}
					-  g\partial_{1}f-       h\partial_{2}f&=     gf\,\Gamma^{1}_{11}  
					+gf\,\Gamma^{ 2}_{12} 
				\end{align}

				\begin{align} \label{eqdiffthet4}
					f   \partial_{2}g& =  - f h\,\Gamma^{1}_{22}-2fg\,\Gamma^{ 3}_{23} 
				\end{align} 
				
				\begin{align} \label{eqdiffthet5}
					-f\partial_{1}g& =  f g\,\Gamma^{1}_{11} +fg\,\Gamma^{ 3}_{13} -hf\,\Gamma^{ 3}_{32}
				\end{align} 
				
				\begin{align} \label{eqdiffthet6}
					- g\partial_{1}g- h\partial_{2}g& =  g^2\,\Gamma^{1}_{11}+ h^2\,\Gamma^{1}_{22}+g^2\,\Gamma^{ 3}_{13}+hg\,\Gamma^{ 3}_{23} 
				\end{align} 
				
				\begin{align} \label{eqdiffthet7}
					f\partial_{2}h&=  -fh\,\Gamma^{ 3}_{23}
				\end{align}
				
				\begin{align} \label{eqdiffthet8}
					- f\partial_{1}h&= 3 fh\,\Gamma^{2}_{12}
				\end{align}
				
				\begin{align} \label{eqdiffthet9}
					- g\partial_{1}h-  h\partial_{2}h&=   3gh\,\Gamma^{2}_{12}  +h^2\,\Gamma^{ 3}_{23}
				\end{align}

				where we defined  the functions $f:=\Theta^{1 2}$, $g:=\Theta^{1 3}$, $h:=\Theta^{2 3}$. First, let us consider the solutions of the choices given in the following table,

				\begin{center}
					\begin{tabular}{ |c|c|c|c| } 
						\hline
						$f$ & $g$ & $h$& Solution \\
						\hline
						{0} & 0 & $h$ &  Eq.  \eqref{eqdiffthet2} $\Rightarrow$ $h=0$ \\ 
						{$f$} & 0 & 0 & $I$\\ 
						{0} & $g$ & $0$ &  Eq.  \eqref{eqdiffthet1}  $\Rightarrow$ $g=0$ \\ 
						$f$& $0$ & $h$&  Eq.  \eqref{eqdiffthet6}  $\Rightarrow$ $h=0$ $\Rightarrow$ $I$ \\
						\hline
					\end{tabular}
				\end{center}
				where   $I$ is  
				the solution of the following differential equations,  
				\begin{align}\label{solI}
					\partial_{2}f&=      0\\\nonumber\\
					\partial_{1}f&=     -  \left(\Gamma^{1}_{11}    +\Gamma^{ 2}_{12} \right)\, f.
				\end{align}
				We are left with three cases \, \newline
				\begin{enumerate} 
					\item \label{case1} $f,g\neq 0$ while $h=0$,\, \newline
					\item \label{case2} $g,h\neq 0$ while $f=0$, \, \newline
					\item \label{case3} $f,g,h\neq 0$.\, \newline
				\end{enumerate} 
				
				For Case  \ref{case1}  we have the following non-vanishing differential equations  
				\begin{align} \label{eqdiffthet21}
					f\partial_{2}f&=      
					-g^2\,\Gamma^{ 2}_{33} 
				\end{align}
				
				\begin{align} \label{eqdiffthet22}
					-    \partial_{1}f&=    f \,\Gamma^{1}_{11}  +f \,\Gamma^{ 2}_{12} 
				\end{align}

				\begin{align} \label{eqdiffthet23}
					\partial_{2}g& =   -2  g\,\Gamma^{ 3}_{23} 
				\end{align} 
				
				\begin{align} \label{eqdiffthet24}
					- \partial_{1}g& =   g\,\Gamma^{1}_{11} + g\,\Gamma^{ 2}_{12}  
				\end{align} 
				These differential equations are solved by a separation ansatz, that is, 
				\begin{align*}
					g(r,\vartheta)=f_1(r)g_2(\vartheta)
				\end{align*}
				where the radial part of both functions is equal, since they satisfy the same differential equation. The solution of Equation \eqref{eqdiffthet23} is given by  
				\begin{align*}
					g(r,\vartheta)= C^2_2 \,f^2_1(r) \,\csc(\vartheta)^2
				\end{align*}
				Inserting this solution  in to  Equation \eqref{eqdiffthet21} after simply rewriting it as 
				\begin{align*}
					\frac{1}{2}\partial_{2}(f^2)&=      
					-g^2\,\Gamma^{ 2}_{33} 
				\end{align*}
				gives us for the function $f$ the following solution 
				\begin{align*}
					f(r,\vartheta)= \sqrt{C_1(r) - C_2^2 f^2_1 (r) \csc(\vartheta)^2}
				\end{align*} 
				Inserting this solution into Equation \eqref{eqdiffthet22} gives us for the function $C_1(r)=C^2_3 f^2_1(r)$  rendering 
				\begin{equation}
					f(r,\vartheta)= C_4 f_1(r)\sqrt{1 - C_5^2 \csc(\vartheta)^2},
				\end{equation}
				where $C_5=C_2^2/C^2_3$, with $C_2,\cdots,C_5\in\R$. To ensure that the solution $f$ is real valued, the constant $C_2=0$ renders $g=0$ and we return to the Solution $I$. 
				
				Taking into account the case \ref{case2}, that is, $f=0$ and $g,h\neq0$ we have the following equations.

				\begin{align}\label{eqdiffthet31}
					0=      h\,\Gamma^{1}_{33}
					-g \,\Gamma^{ 2}_{33} 
				\end{align}

				\begin{align}\label{eqdiffthet32}
					- g\partial_{1}g- h\partial_{2}g& =  g^2\,\Gamma^{1}_{11}+ h^2\,\Gamma^{1}_{22}+g^2\,\Gamma^{ 3}_{13}+hg\,\Gamma^{ 3}_{23} 
				\end{align} 
				
				\begin{align}\label{eqdiffthet33}
					- g\partial_{1}h-  h\partial_{2}h&=   3gh\,\Gamma^{2}_{12}  +h^2\,\Gamma^{ 3}_{23}
				\end{align} 
				From Equation \eqref{eqdiffthet31} we obtain 
				\begin{align*}
					h=     
					g \,\Gamma^{ 2}_{33} / \Gamma^{1}_{33}=:b(r,\vartheta) g 
				\end{align*}
				where $ b(r,\vartheta) :=\exp(- 2 \alpha(r) ) \frac{\cot{\vartheta}}{r} $.
				Plugging this into the other two differential equations gives us 
				
				\begin{align}\label{eqdiffthet34}
					- g\partial_{1}g- bg\partial_{2}g& =  g^2\,\Gamma^{1}_{11}+  g^2\,\Gamma^{ 3}_{13} 
					\\&=\left( \dfrac{1}{r}-\partial_1\alpha(r)\right)   g^2  \end{align} 
				where we used the fact that $b^2\,\Gamma^{1}_{22}+b\,\Gamma^{ 3}_{23} =0$.

				\begin{align}\label{eqdiffthet35}
					-  g\partial_{1}g  
					-  b g\partial_{2}g &=   3 g^2\,\Gamma^{2}_{12}  - g^2\exp(- 2 \alpha(r) ) \frac{1}{r}+  g^2\left(-2\partial_1\alpha(r) -\frac{1}{r}\right)     
					\\&=\left(\frac{2}{r}  -\exp(- 2 \alpha(r)) \frac{1}{r}-2\partial_1\alpha(r))\right)      g^2
				\end{align} where in the last lines we used 
				the relation $b \,\Gamma^{ 3}_{23}+    \partial_{2}b=  -\exp(- 2 \alpha(r) ) \frac{1}{r}$   which follows from 
				the derivatives of the function $b$ that  are given by 
				\begin{align*}
					\partial_1 b(r,\vartheta)&= -2\partial_1\alpha(r)  \exp(- 2 \alpha(r) ) \frac{\cot{\vartheta}}{r} 
					-\exp(- 2 \alpha(r) ) \frac{\cot{\vartheta}}{r^2} \\&=
					\left(-2\partial_1\alpha(r) -\frac{1}{r}\right)b(r,\vartheta),\\ 
					\partial_2 b(r,\vartheta)&=-\csc(\vartheta)^2\exp(- 2 \alpha(r) ) \frac{1}{r}  .
				\end{align*}
				Subtracting the   equations  \eqref{eqdiffthet34} and \eqref{eqdiffthet35} renders 
				
				\begin{align*}
					0=  \partial_1\alpha(r)-  \dfrac{1}{r} + \exp(- 2 \alpha(r) ) \frac{1}{r} 
				\end{align*}
				The solution is given by 
				\begin{align*}
					\exp(- 2 \alpha(r) )=  {1-C_1e^{z_1}r^2},
				\end{align*}
				where $z_1\in\C:e^{z_1}\in\R$ and $C_1$ is a constant of spatial dimension $-2$. Therefore, there is only a solution for a certain class of $\alpha(r)$. Since the differential equations for the function $g$ do not depend on the angel $\vartheta$ we can assume $\partial_2 g=0$ which renders the differential equation \eqref{eqdiffthet34}
				\begin{align}
					-  \partial_{1}g  & =  \left( \dfrac{1}{r}-\partial_1\alpha(r)\right)   g ,
				\end{align}
				which has for the certain class of solutions w.r.t.\ $\alpha$ the following form 
				\begin{align*}
					g(r)=\exp(  \alpha(r) ) \frac{1}{r} .
				\end{align*}
				This solution gives us thus the following for the function $h$, 
				\begin{align*}
					h(r,\vartheta)= \exp(-   \alpha(r) ) \frac{\cot{\vartheta}}{r^2} .
				\end{align*}
				
				The last set of solutions is given by Case \ref{case3}, namely assuming that no function vanishes. For this case, we have,  
				
				\begin{align} \label{eqdiffthet41}
					f\partial_{2}f&=     gh\,\Gamma^{1}_{33}
					-g^2\,\Gamma^{ 2}_{33} 
				\end{align}
				
				\begin{align} \label{eqdiffthet42}
					-   f\partial_{1}f&=    f^2\,\Gamma^{1}_{11} + h^2\,\Gamma^{1}_{33}
					-hg\,\Gamma^{ 2}_{33}+f^2\,\Gamma^{ 2}_{12} 
				\end{align}
				
				\begin{align} \label{eqdiffthet43}
					-  g\partial_{1}f-       h\partial_{2}f&=     gf\,\Gamma^{1}_{11}  
					+gf\,\Gamma^{ 2}_{12} 
				\end{align}

				\begin{align} \label{eqdiffthet44}
					\partial_{2}g& =  -  h\,\Gamma^{1}_{22}-2g\,\Gamma^{ 3}_{23} 
				\end{align} 
				
				\begin{align} \label{eqdiffthet45}
					-\partial_{1}g& =   g\,\Gamma^{1}_{11} +g\,\Gamma^{ 3}_{13} -h\,\Gamma^{ 3}_{32}
				\end{align}  
				
				\begin{align} \label{eqdiffthet47}
					\partial_{2}h&=  -h\,\Gamma^{ 3}_{23}
				\end{align}
				
				\begin{align} \label{eqdiffthet48}
					- \partial_{1}h&= 3 h\,\Gamma^{2}_{12}
				\end{align}
				The strategy here is to first solve $h$ by a separation of variables ansatz and then use the solution to solve for $g$. Finally, plugging the functions $g$ and $h$ into the differential equations for $f$  renders the last differential equations.
				Therefore, we first solve for $h$
				\begin{align*}
					h(r,\vartheta)=\frac{C_1}{r^3} \csc(\vartheta)
				\end{align*}
				Using this to solve Equation \eqref{eqdiffthet44} we obtain
				\begin{align*}
					g(r,\vartheta)= C_1 \frac{e^{2\alpha(r)}}{r^2}\cot(\vartheta) \csc(\vartheta) + C_3(r) \csc^2(\vartheta)  
				\end{align*}
				Plugging this into differential equation \eqref{eqdiffthet45} gives us a differential equation for $C_3(r)$
				\begin{align}
					\partial_1C_3(r)= -C_1\cos(\vartheta)\dfrac{e^{2\alpha(r)}}{r^2}  \partial_1\alpha+  C_1\cos(\vartheta)\dfrac{e^{2\alpha(r)}}{r^3}+C_3(r)\partial_1\alpha-\frac{C_3(r)}{r}+C_1 \dfrac{\cos(\vartheta)}{r^3}
				\end{align}
				Since the function $C_3(r)$ has to be independent of the angle $\vartheta$ we take the derivative w.r.t.\ this angle which results  in 
				\begin{align*}
					- r {e^{2\alpha(r)}}   \partial_1\alpha(r)+   {e^{2\alpha(r)}} +1  =0
				\end{align*}
				This therefore results as with Case \ref{case2} into a differential equation of $\alpha(r)$. The solution in case $C_1\neq0$ is given by 
				\begin{align*}
					\alpha(r) = 1/2 \log(C_4r^2 - 1).
				\end{align*}
				However, since $r$ is unbounded  the solution  $C_1=0$ reduces Case \ref{case3} to Case \ref{case1}, i.e.\ Solution $I$, see  Equation \eqref{solI}.

			\end{proof}
			\subsection{Proof of Result \ref{propfrwl}}\label{appropfrwl}
			\begin{proof}

				The only non-vanishing Christoffel symbols for the FRWL spacetime in case of flat spatial geometry are given by 
				\begin{align}
					\Gamma^0_{ij}=\delta_{ij} \,a\,\dot{a}, \qquad \qquad \Gamma^i_{0j}=\delta^{i}_{\,\,j} {\dot{a}}/{a}
				\end{align}
				We have the following equation that follows from Condition \eqref{eq:equiv}
				\begin{align*}
					\Theta^{\mu\nu}\partial_{\nu}\Theta^{\alpha\beta}= - \Theta^{\mu\nu}\Theta^{\lambda\beta}\,\Gamma^{\alpha}_{\nu\lambda}-\Theta^{\mu\nu}\Theta^{\alpha\lambda}\,\Gamma^{\beta}_{\nu\lambda}
				\end{align*}
				Assuming that the Poisson tensor does not depend on the spatial coordinates reduces the differential equations to 
				\begin{align*}
					\Theta^{\mu0}\partial_{0}\Theta^{\alpha\beta}= - \Theta^{\mu\nu}\Theta^{\lambda\beta}\,\Gamma^{\alpha}_{\nu\lambda}-\Theta^{\mu\nu}\Theta^{\alpha\lambda}\,\Gamma^{\beta}_{\nu\lambda}
				\end{align*}
				For $\mu=0$ we have 
				\begin{align*}
					0 = - \Theta^{0i}\Theta^{\lambda\beta}\,\Gamma^{\alpha}_{i\lambda}-\Theta^{0i}\Theta^{\alpha\lambda}\,\Gamma^{\beta}_{i\lambda}
				\end{align*}
				choosing $\alpha=0$, $\beta=j$ renders 
				
				\begin{align*}
					0 &= - \Theta^{0i}\Theta^{\lambda j}\,\Gamma^{0}_{i\lambda}-\Theta^{0i} \Theta^{0\lambda}\,\Gamma^{j}_{i\lambda}
					\\&=
					- \Theta^{0i}\Theta^{k j}\,\Gamma^{0}_{ik}
				\end{align*}
				Next, we choose $\alpha=k$, $\beta=j$ renders 
				\begin{align*}
					0 &= - \Theta^{0i}\Theta^{\lambda j}\,\Gamma^{k}_{i\lambda}-\Theta^{0i}\Theta^{k\lambda}\,\Gamma^{j}_{i\lambda}
					\\&= - \Theta^{0i}\Theta^{0 j}\,\Gamma^{k}_{i0}-\Theta^{0i}\Theta^{k0}\,\Gamma^{j}_{i0}
					\\&=   - \Theta^{0k}\Theta^{0 j}  \,{\dot{a}}/{a}-\Theta^{0j}\Theta^{k0} \,{\dot{a}}/{a}
				\end{align*}
				which is automatically fulfilled since $\Theta^{0j}=-\Theta^{j0}$.
				For $\mu=l$ we have 
				
				\begin{align*}
					\Theta^{l 0}\partial_{0}\Theta^{\alpha\beta}= - \Theta^{l\nu}\Theta^{\lambda\beta}\,\Gamma^{\alpha}_{\nu\lambda}-\Theta^{l\nu}\Theta^{\alpha\lambda}\,\Gamma^{\beta}_{\nu\lambda}
				\end{align*}
				choosing $\alpha=0$, $\beta=j$ reduces to 
				\begin{align*}
					\Theta^{l 0}\partial_{0}\Theta^{0j}&= - \Theta^{l\nu}\Theta^{\lambda j}\,\Gamma^{0}_{\nu\lambda}-\Theta^{l\nu}\Theta^{0\lambda}\,\Gamma^{j}_{\nu\lambda}
					\\&= - \Theta^{li}\Theta^{k j}\,\Gamma^{0}_{ik}-\Theta^{l0}\Theta^{0k}\,\Gamma^{j}_{0k}
					\\&= - \Theta^{li}\Theta^{k j}\,\Gamma^{0}_{ik}-\Theta^{l0}\Theta^{0j  }\, {\dot{a}}/{a}
				\end{align*}
				and as last equations for $\alpha=k$, $\beta=j$ we have
				
				\begin{align*}
					\Theta^{l 0}\partial_{0}\Theta^{kj}&= - \Theta^{l\nu}\Theta^{\lambda j}\,\Gamma^{k}_{\nu\lambda}-\Theta^{l\nu}\Theta^{k\lambda}\,\Gamma^{j}_{\nu\lambda}\\&= - \Theta^{l0}\Theta^{ij}\,\Gamma^{k}_{0i}- \Theta^{li}\Theta^{0 j}\,\Gamma^{k}_{i0}
					- \Theta^{l0}\Theta^{ik}\,\Gamma^{j}_{0i}- \Theta^{li}\Theta^{0 k}\,\Gamma^{j}_{i0}
				\end{align*}
				Assuming that $\Theta^{ik}=0$ the equations reduce to 
				\begin{align*}
					\partial_{0}\Theta^{0j}&=  - \Theta^{0j  }\, {\dot{a}}/{a}.
				\end{align*}
				
			\end{proof}

			\bibliographystyle{alpha}
			\bibliography{allliterature1b.bib}

\newcommand{\etalchar}[1]{$^{#1}$}
\begin{thebibliography}{MMMZ04}

\bibitem[Ahl94]{Ahluwalia1993}
Dharam~Vir Ahluwalia.
\newblock {Quantum measurements, gravitation, and locality}.
\newblock {\em Phys.~Lett.~B}, 339:301--303, 1994.

\bibitem[ALV08]{pasc}
Paolo Aschieri, Fedele Lizzi, and Patrizia Vitale.
\newblock Twisting all the way: From classical mechanics to quantum fields.
\newblock {\em Physical Review D}, 77(2), January 2008.

\bibitem[BF14]{bf}
Marcos Brum and Klaus Fredenhagen.
\newblock 'vacuum-like' hadamard states for quantum fields on curved spacetimes.
\newblock {\em Classical and Quantum Gravity}, 31(2):025024, 2014.

\bibitem[BFF{\etalchar{+}}77]{Bayen1}
F.~Bayen, M.~Flato, C.~Fronsdal, A.~Lichnerowicz, and D.~Sternheimer.
\newblock {Quantum Mechanics as a Deformation of Classical Mechanics}.
\newblock {\em Lett. Math. Phys.}, 1:521--530, 1977.

\bibitem[BFF{\etalchar{+}}78a]{Bayen3}
F.~Bayen, M.~Flato, C.~Fronsdal, A.~Lichnerowicz, and D.~Sternheimer.
\newblock {Deformation Theory and Quantization. 2. Physical Applications}.
\newblock {\em Annals Phys.}, 111:111, 1978.

\bibitem[BFF{\etalchar{+}}78b]{DEFBay}
F.~{Bayen}, M.~{Flato}, C.~{Fronsdal}, A.~{Lichnerowicz}, and D.~{Sternheimer}.
\newblock {Deformation theory and quantization. I. Deformations of symplectic structures}.
\newblock {\em Annals of Physics}, 111(1):61--110, March 1978.

\bibitem[BLS11]{BLS}
Detlev Buchholz, Gandalf Lechner, and Stephen~J. Summers.
\newblock {Warped Convolutions, Rieffel Deformations and the Construction of Quantum Field Theories}.
\newblock {\em Commun. Math. Phys.}, 304:95--123, 2011.

\bibitem[BM17]{majbeg}
Edwin~J. Beggs and Shahn Majid.
\newblock Poisson–riemannian geometry.
\newblock {\em Journal of Geometry and Physics}, 114:450--491, 2017.

\bibitem[Bou03]{boum}
Mohamed Boucetta.
\newblock Riemann–poisson manifolds and kähler–riemann foliations.
\newblock {\em Comptes Rendus Mathematique}, 336(5):423--428, 2003.

\bibitem[Cha01]{bf0}
Ali~H. Chamseddine.
\newblock Deforming einstein's gravity.
\newblock {\em Physics Letters B}, 504(1):33--37, 2001.

\bibitem[CHM08]{Crisp}
Luís C.~B. Crispino, Atsushi Higuchi, and George E.~A. Matsas.
\newblock The unruh effect and its applications.
\newblock {\em Reviews of Modern Physics}, 80(3):787–838, July 2008.

\bibitem[CTSZ08]{bf1}
M.~Chaichian, A.~Tureanu, M.R. Setare, and G.~Zet.
\newblock On black holes and cosmological constant in noncommutative gauge theory of gravity.
\newblock {\em Journal of High Energy Physics}, 2008(04):064, apr 2008.

\bibitem[CTZ08]{bf2}
M.~Chaichian, A.~Tureanu, and G.~Zet.
\newblock Corrections to schwarzschild solution in noncommutative gauge theory of gravity.
\newblock {\em Physics Letters B}, 660(5):573--578, 2008.

\bibitem[Dav75]{INN8}
P~C~W Davies.
\newblock Scalar production in schwarzschild and rindler metrics.
\newblock {\em Journal of Physics A: Mathematical and General}, 8(4):609, apr 1975.

\bibitem[DFR95]{DFR}
Sergio Doplicher, Klaus Fredenhagen, and John~E. Roberts.
\newblock {The Quantum structure of space-time at the Planck scale and quantum fields}.
\newblock {\em Commun. Math. Phys.}, 172:187--220, 1995.

\bibitem[Dri88]{drin}
V.~G. Drinfel'd.
\newblock Quantum groups.
\newblock {\em Journal of Soviet Mathematics}, 41(2):898--915, 1988.

\bibitem[Fed94]{fedo1}
Boris~V. Fedosov.
\newblock {A simple geometrical construction of deformation quantization}.
\newblock {\em Journal of Differential Geometry}, 40(2):213 -- 238, 1994.

\bibitem[FM21]{MF}
M.~B. Fröb and A.~Much.
\newblock Strict deformations of quantum field theory in de sitter spacetime.
\newblock {\em Journal of Mathematical Physics}, 62(6):062302, June 2021.

\bibitem[Ful73]{INN7}
Stephen~A. Fulling.
\newblock Nonuniqueness of canonical field quantization in riemannian space-time.
\newblock {\em Phys. Rev. D}, 7:2850--2862, May 1973.

\bibitem[GKP96]{GRFSPH3}
H.~Grosse, C.~Klimcik, and P.~Presnajder.
\newblock {Towards finite quantum field theory in noncommutative geometry}.
\newblock {\em Int. J. Theor. Phys.}, 35:231--244, 1996.

\bibitem[GL07]{GL1}
Harald Grosse and Gandalf Lechner.
\newblock {Wedge-Local Quantum Fields and Noncommutative Minkowski Space}.
\newblock {\em JHEP}, 0711:012, 2007.

\bibitem[GNV14]{bigbang3}
M.~Gorji, K.~Nozari, and B.~Vakili.
\newblock Spacetime singularity resolution in snyder noncommutative space.
\newblock {\em Physical Review D}, 89(8), April 2014.

\bibitem[GP93]{GRFSPH1}
H.~Grosse and P.~Presnajder.
\newblock {The Construction on noncommutative manifolds using coherent states}.
\newblock {\em Lett. Math. Phys.}, 28:239--250, 1993.

\bibitem[GP95]{GRFSPH2}
H.~Grosse and P.~Presnajder.
\newblock {The Dirac operator on the fuzzy sphere}.
\newblock {\em Lett. Math. Phys.}, 33:171--182, 1995.

\bibitem[GSW08]{eg3}
H.~Grosse, H.~Steinacker, and M.~Wohlgenannt.
\newblock Emergent gravity, matrix models and uv/ir mixing.
\newblock {\em Journal of High Energy Physics}, 2008(04):023–023, April 2008.

\bibitem[Hac10]{semiclas4}
Thomas-Paul Hack.
\newblock On the backreaction of scalar and spinor quantum fields in curved spacetimes - from the basic foundations to cosmological applications, 2010.

\bibitem[Haw04]{hawk1}
Eli Hawkins.
\newblock Noncommutative rigidity.
\newblock {\em Communications in Mathematical Physics}, 246(2):211--235, 2004.

\bibitem[Haw07]{hawk2}
Eli Hawkins.
\newblock {The structure of noncommutative deformations}.
\newblock {\em Journal of Differential Geometry}, 77(3):385 -- 424, 2007.

\bibitem[JA21]{semiclas1}
Benito~A. Juárez-Aubry.
\newblock Semiclassical gravity in static spacetimes as a constrained initial value problem, 2021.

\bibitem[Kon03]{Kont}
Maxim Kontsevich.
\newblock Deformation quantization of poisson manifolds.
\newblock {\em Letters in Mathematical Physics}, 66(3):157–216, December 2003.

\bibitem[Lee18]{LeeRiem}
John~M Lee.
\newblock {\em Introduction to Riemannian manifolds (Corrected version of second edition)}.
\newblock Graduate texts in mathematics 176. Springer Nature, 2 edition, 2018.

\bibitem[Mad99]{madore}
J.~Madore.
\newblock {\em An Introduction to Noncommutative Differential Geometry and its Physical Applications}.
\newblock London Mathematical Society lecture note series 257. Cambridge University Press, 2nd ed edition, 1999.

\bibitem[MG09]{muellercat}
Thomas {Mueller} and Frank {Grave}.
\newblock {Catalogue of Spacetimes}.
\newblock {\em arXiv e-prints}, page arXiv:0904.4184, April 2009.

\bibitem[MMMZ04]{bigbang1}
M.~Maceda, J.~Madore, P.~Manousselis, and G.~Zoupanos.
\newblock Can non-commutativity resolve the big-bang singularity?
\newblock {\em The European Physical Journal C}, 36(4):529–534, August 2004.

\bibitem[MT88]{morwh}
Michael~S. Morris and Kip~S. Thorne.
\newblock {Wormholes in spacetime and their use for interstellar travel: A tool for teaching general relativity}.
\newblock {\em American Journal of Physics}, 56(5):395--412, 05 1988.

\bibitem[{Muc}21]{muchboumaye}
Albert {Much}.
\newblock {A Deformation Quantization for Non-Flat Spacetimes and Applications to QFT}.
\newblock {\em arXiv e-prints}, page arXiv:2109.14507, September 2021.

\bibitem[Mü01]{mueck}
Wolfgang Mück.
\newblock Ideas on the semi-classical path integral over embedded manifolds.
\newblock {\em Fortschritte der Physik}, 49(4-6):607--615, 2001.

\bibitem[Pin11]{semiclas3}
Nicola Pinamonti.
\newblock On the initial conditions and solutions of the semiclassical einstein equations in a cosmological scenario.
\newblock {\em Communications in Mathematical Physics}, 305(3):563--604, 2011.

\bibitem[Sie15]{semiclas2}
Daniel Siemssen.
\newblock The semiclassical einstein equation on cosmological spacetimes, 2015.

\bibitem[Ste07]{eg2}
Harold Steinacker.
\newblock Emergent gravity from noncommutative gauge theory.
\newblock {\em Journal of High Energy Physics}, 2007(12):049, dec 2007.

\bibitem[Ste18]{bigbang4}
Harold~C. Steinacker.
\newblock Cosmological space-times with resolved big bang in yang-mills matrix models.
\newblock {\em Journal of High Energy Physics}, 2018(2), February 2018.

\bibitem[Sza03]{SZNCQFT}
R~Szabo.
\newblock Quantum field theory on noncommutative spaces.
\newblock {\em Physics Reports}, 378(4):207–299, May 2003.

\bibitem[TS24]{bf4}
Abdellah Touati and Zaim Slimane.
\newblock Quantum tunneling from schwarzschild black hole in non-commutative gauge theory of gravity.
\newblock {\em Physics Letters B}, 848:138335, 2024.

\bibitem[TV14]{bigbang2}
Luca Tomassini and Stefano Viaggiu.
\newblock Building non-commutative spacetimes at the planck length for friedmann flat cosmologies.
\newblock {\em Classical and Quantum Gravity}, 31(18):185001, August 2014.

\bibitem[{Unr}76]{INN6}
W.~G. {Unruh}.
\newblock {Notes on black-hole evaporation}.
\newblock {\em Phys. Rev. D}, 14(4):870--892, August 1976.

\bibitem[Vai94]{poison}
Izu Vaisman.
\newblock {\em Lectures on the Geometry of Poisson Manifolds}.
\newblock Progress in Mathematics No. 118. Birkhäuser, 1 edition, 1994.

\bibitem[Wal07]{waldpoiss}
S.~Waldmann.
\newblock {\em Poisson-Geometrie und Deformationsquantisierung: Eine Einf{\"u}hrung}.
\newblock Masterclass. Springer Berlin Heidelberg, 2007.

\bibitem[Wal10]{WA}
Robert~M. Wald.
\newblock {\em General Relativity}.
\newblock University of Chicago Press, 2010.

\bibitem[Yan09]{eg1}
Hyun~Seok Yang.
\newblock Emergent gravity from noncommutative spacetime.
\newblock {\em International Journal of Modern Physics A}, 24(24):4473–4517, September 2009.

\end{thebibliography}

		\end{document}